\documentclass{ws-ijfcs}
\usepackage{enumerate}
\usepackage{url}
\usepackage{latexsym}
\usepackage{amssymb}
\usepackage{tikz}
\usetikzlibrary{calc}
\usepackage{graphicx}
\usepackage{amsmath}
\usepackage{bbm}
\usepackage{mathrsfs}
\usepackage{cite}
\usepackage{wrapfig}
\usepackage{verbatim}
\usepackage{multirow}

\newtheorem{property}{Property}

\newcommand{\dist}{\operatorname{dist}}
\newcommand{\fib}{\mbox{\upshape\sf fib}}
\newcommand{\Omicron}{O}
\newcommand{\bcdot}{$\discretionary{\mbox{$ \cdot $}}{}{}$}

\begin{document}
\catchline{}{}{}{}{}

\title{Lower Bounds for Synchronizing Word Lengths in Partial 
Automata}
\author{Michiel de Bondt}
\address{Department of Computer Science, Radboud University Nijmegen, The Netherlands \\
\email{m.debondt@math.ru.nl}}
\author{Henk Don}
\address{Department of Mathematics, Radboud University Nijmegen, The Netherlands \\
\email{h.don@math.ru.nl}}
\author{Hans Zantema}
\address{Department of Computer Science, TU Eindhoven, The Netherlands, and \\
Department of Computer Science, Radboud University Nijmegen, The Netherlands \\
\email{h.zantema@tue.nl}}

\maketitle
\setcounter{footnote}{0}

\begin{abstract}
It was conjectured by \v{C}ern\'y in 1964, that a synchronizing DFA on $n$ states always has a synchronizing word of length at most $(n-1)^2$, and he gave a sequence of DFAs for which
this bound is reached. Until now a full analysis of all DFAs reaching this bound was only given for $n
\leq 5$, and with bounds on the number of symbols for $n \leq 12$. Here we give the full analysis for $n
\leq 7$, without bounds on the number of symbols.

For PFAs (partial automata) on $\leq 7$ states we do a similar analysis as for DFAs and
find the maximal shortest synchronizing word lengths, exceeding $(n-1)^2$ for $n \geq 4$. Where DFAs with long synchronization typically have very few symbols, for PFAs we observe that more symbols may increase the synchronizing word length. For PFAs on $\leq 10$ states and two symbols we investigate all occurring synchronizing word lengths.

We give series of PFAs on two and three symbols, reaching the maximal possible length for some small values of $n$. For $n=6,7,8,9$, the construction on two symbols is the unique one reaching the maximal length. For both series the growth is faster than $(n-1)^2$, although still quadratic.

Based on string rewriting, for arbitrary size
we construct a PFA on three symbols with exponential shortest synchronizing word
length, giving significantly better bounds than earlier exponential constructions. We give a transformation
of this PFA to a PFA on two symbols keeping exponential shortest synchronizing word length, yielding a
better bound than applying a similar known transformation. Both PFAs are transitive.

Finally, we show that exponential lengths are even possible with just one single undefined transition, again with transitive constructions.
\end{abstract}

\section{Introduction and Preliminaries}

A {\em deterministic finite automaton (DFA)} over a finite alphabet $\Sigma$ is called
{\em synchronizing}, if it admits a {\em synchronizing word}. A word $w \in \Sigma^*$ is
called {\em synchronizing} (or directed, or reset), if, starting in any state $q$, after
reading $w$, one always ends in one particular state $q_s$. So reading $w$ acts as a reset
button:
no matter in which state the system is, it always moves to the particular state $q_s$.
Now \v{C}ern\'y's conjecture \cite{C64} states:
\begin{quote}
Every synchronizing DFA on $n$ states admits a synchronizing word of length $\leq (n-1)^2$.
\end{quote}

Surprisingly, despite extensive effort, this conjecture is still open, and even the best known upper
bounds are still cubic in $n$. In 1983 Pin \cite{pin} established the bound 
$\frac{1}{6}(n^3-n)-1$ for $n>4$, based on \cite{frankl} and Pin's thesis. 
Recently, a slight asymptotic improvement to Pin's bound has been obtained by Szyku\l{}a 
\cite{szykula18} (effective for $n \ge 724$). For a survey on synchronizing automata and 
\v{C}ern\'y's conjecture, we refer to 
\cite{volkov}.

Formally, a {\em deterministic finite automaton (DFA)} over a finite alphabet $\Sigma$ consists of a finite set
$Q$ of states and a map $\delta: Q \times \Sigma \to Q$.\footnote{For synchronization the
initial state and the set of final states in the standard definition may be ignored.}
For $w \in \Sigma^*$ and $q \in Q$, we define $qw$ inductively by
$q \lambda = q$ and $qaw = \delta(q,a)w$ for
$a \in \Sigma$, where $\lambda$ is the empty word.
So $qw$ is the state where one ends, when starting in $q$ and reading the symbols in $w$ consecutively, and $qa$ is a short hand notation for
$\delta(q,a)$. A word $w \in \Sigma^*$ is called {\em synchronizing}, if a
state $q_s \in Q$ exists such that $q w = q_s$ for all $q \in Q$.

In \cite{C64}, \v{C}ern\'y already gave DFAs for which the bound of the conjecture is attained:
for $n \geq 2$ the DFA $C_n$ is defined
to consist of $n$ states $1,2,\ldots,n$, and two symbols $a,b$, acting by
$qa = q+1$ for $q = 1,\ldots,n-1$, $\delta(n,a) = 1$, and  $qb = q$ for
$q = 2,\ldots,n$, $1b = 2$.

\begin{wrapfigure}{r}{4cm}
\begin{tikzpicture}
\useasboundingbox (-1,-0.5) rectangle (3,1.8);
\node[circle,draw,inner sep=0pt,minimum width=5mm] (1) at (0,2) {$1$};
\node[circle,draw,inner sep=0pt,minimum width=5mm] (2) at (2,2) {$2$};
\node[circle,draw,inner sep=0pt,minimum width=5mm] (3) at (2,0) {$3$};
\node[circle,draw,inner sep=0pt,minimum width=5mm] (4) at (0,0) {$4$};
\draw[->] (1) -- node[above,inner sep=1pt] {$a,b$} (2);
\draw[->] (2) -- node[right,inner sep=2pt] {$a$} (3);
\draw[->] (3) -- node[below,inner sep=3pt] {$a$} (4);
\draw[->] (4) -- node[left,inner sep=2pt] {$a$} (1);
\draw[->] (2) edge[out=90,in=0,looseness=7] node[right,inner sep=4pt] {$b$} (2);
\draw[->] (3) edge[out=-90,in=0,looseness=7] node[right,inner sep=4pt] {$b$} (3);
\draw[->] (4) edge[out=-90,in=-180,looseness=7] node[left,inner sep=4pt] {$b$} (4);
\node at (1,-0.9) {$C_4$};
\end{tikzpicture}
\end{wrapfigure}

For $n=4$, this is depicted on the right.
For $C_n$, the string $w = b (a^{n-1}b)^{n-2}$ of length $|w| = (n-1)^2$ satisfies
$qw = 2$ for all $q
\in Q$, so $w$ is synchronizing. No shorter synchronizing word exists for
$C_n$, as is shown in \cite{C64}, showing that the bound in \v{C}ern\'y's conjecture is sharp.

One goal of this paper is to investigate the synchronizing word lengths of all DFAs on at most $7$ states. 
We also search for the maximal word lengths when restricting to DFAs with a given alphabet size. 
The main result on DFAs is that \v{C}ern\'y's conjecture is true for $n\leq 7$. Our results extend
those in \cite{10.1007/978-3-319-40946-7_15}, in which \v{C}ern\'y's conjecture is verified for 
$n\leq 5$. A complete analysis of all DFAs of 
$n = 6$ and $n = 7$ states is not provided in \cite{10.1007/978-3-319-40946-7_15}: the number of symbols is limited to $6$ and $4$ respectively.
The computations in
\cite{10.1007/978-3-319-40946-7_15} extend several results by the same authors.

A generalization of a DFA is a {\em Partial Finite Automaton} (PFA); the only difference is that now the
transition function $\delta$ is allowed to be partial. In a PFA, $qw$ may be undefined, in fact it is only defined if every step is
defined. A word $w \in \Sigma^*$ is called {\em carefully synchronizing} for a PFA, if a
state $q_s \in Q$ exists such that $q w$ is defined and  $q w = q_s$ for all $q \in Q$. In other words:
starting in any state $q$ and reading $w$, every step is defined and one always ends in state $q_s$. A PFA, in particular a DFA, is called \emph{transitive} or \emph{strongly connected} if for every ordered pair $(q,q')$ of states, there is a word $w\in\Sigma^\star$ such that $qw=q'$.
As being a generalization of DFAs, the shortest carefully synchronizing word may be longer. For all $n \geq 4$
we show that this is indeed the case: for $n=4,5,6,7$ we find the maximal shortest carefully synchronizing word length
to be 10, 21, 37 and 63, respectively. 

Also for PFAs we investigate the dependence on the alphabet size. 
To exclude infinitely many trivial extensions, we only consider \emph{basic} PFAs: 
no two symbols act in the same way, no symbol acts as the identity and no symbol is a 
restricted version of either another symbol or the identity. 
Obviously, these properties have no influence on synchronization. 
Somewhat surprisingly, we find that larger alphabets may lead to longer 
carefully synchronizing words, in contrast to the situation for DFAs. 

We compute all binary PFAs with up to $10$ states, to obtain all possible 
synchronization lengths, both for DFAs and proper PFAs. For DFAs, the authors of 
\cite{10.1007/978-3-319-40946-7_15} got as far as $12$ states with obtaining
these lenghts. With that, they extended the maximum number of states from
earlier analyses, by themselves and by others, ranging from $9$ states to $11$
states. The authors of \cite{10.1007/978-3-319-40946-7_15} obtained all possible 
synchronization lengths for ternary DFAs with $8$ states as well. Several gaps 
exist in the ranges of synchronization lengths for binary DFAs. It appears 
that such gaps also exist for binary PFAs.

For every $n$ we give a PFA on $n$ states and 2 symbols for which we exactly compute the
shortest carefully synchronizing word length, for every $n \geq 6$ strictly exceeding $(n-1)^2$. 
This length is quadratic in $n$, but it is not a polynomial: the precise formula deals with 
Fibonacci numbers. For $n=6,7,8,9$ this is the only construction giving the maximal shortest 
synchronizing word length for binary PFAs. Similarly, we give a sequence of PFAs on three symbols, 
reaching the maximal length for $n=3,4,5$. 

For PFAs the maximal length grows exponentially in $n$, as was already observed by Rystsov
\cite{rystsov}. Rystsov established the lower bound $\Omega((3-\varepsilon)^{n/3})$ and 
the upper bound $O((3+\varepsilon)^{n/3})$. The upper bound can be found in \cite{MR3746507} as well.
Martyugin \cite{M10} established the lower bound $\Omega(3^{n/3})$ with a construction in which the number 
of symbols is linear in $n$. 

In \cite{Mar08}, the author Martyugin obtained a lower bound for the synchronization of PFAs with
a constant alphabet size, which lies between polynomial and exponential, as a result of an elegant
construction of a series of PFAs (see also the last section of \cite{debondt18}). 
In \cite{10.1007/978-3-642-38536-0_7}, the same author obtained a near-exponential lower bound, 
using a different construction of PFAs. 
In \cite{V16} it was shown that exponential bounds exist for every constant alphabet
size being at least two. For two symbols the bound $\Omega(2^{n/36})$ was given for the transitive case 
and the bound $\Omega(2^{n/26})$ for the general case. Our construction strongly improves 
this and gives length $\Omega(2^{n/3-3\log_2(n)/2})=\Omega(2^{n/3}/n^{3/2})$ for binary PFAs 
and length $\Omega(2^{2n/5-\log_2(n)})=\Omega(2^{2n/5}/n)$ for ternary PFAs, 
with transitive constructions. Finally we show that both in the binary and ternary case 
exponential growth is even possible with a single undefined transition, again with 
transitive constructions. For the ternary construction, the key idea is that synchronization is 
forced to mimic exponentially many string rewrite steps, similar to binary counting. 
This ternary PFA can be transformed to a binary PFA by a
standard construction for which we develop a substantial improvement.

The decision problems which correspond to 
our asymptotic constructions are PSPACE-complete, if we do not take transitivity into account.
This follows from \cite{10.1007/978-3-642-31606-7_24}, in which the most specific decision
problem is treated, namely the problem of determining if a binary PFA with only one undefined 
transition is carefully synchronizing. The fact that this problem is PSPACE-complete
already suggested the existence of a nonpolynomial construction, because otherwise 
we would have had PSPACE = NP. However, the construction in \cite{10.1007/978-3-642-31606-7_24} is 
not transitive. Using \cite[Lemma 2]{V16} and \cite[Lemma 6]{V16}, one can make the 
construction transitive, but the property of having only one undefined transition will
be affected. So if we do take transitivity into account, then PSPACE-completeness is 
obtained for the decision problems which correspond to our asymptotic constructions, 
except the last one (with only one undefined transition).

The basic tool to analyze (careful) synchronization is the {\em power automaton}.
For any DFA or PFA $(Q,\Sigma, \delta)$, its power automaton is the DFA $(2^Q,\Sigma, \delta')$
where $\delta' : 2^Q \times \Sigma \to 2^Q$ is defined by
$\delta'(V,a) = \{q \in Q \mid \exists p \in V : \delta(p,a) = q \}$, if $\delta(p,a)$ is defined for
all $p \in V$, otherwise $\delta'(V,a) = \emptyset$.
For any $V \subseteq Q, w \in \Sigma^*$, we define $Vw$ as before, using $\delta'$ instead of
$\delta$.  From this definition, one easily proves that $Vw = \{ qw \mid q \in V \}$ if $qw$ is defined for
all $q \in V$, otherwise $Vw = \emptyset$, for any $V \subseteq Q, w \in \Sigma^*$.
A set of the shape $\{q\}$ for $q \in Q$ is called a {\em singleton}.
So a word $w$ is (carefully) synchronizing, if and only if $Qw$ is a singleton.
Hence a DFA (PFA) is (carefully) synchronizing, if and only if its
power automaton admits a path from $Q$ to a singleton, and the shortest
length of such a path corresponds to the shortest length of a (carefully) synchronizing word.

This paper is an extended version of the DLT2017 paper \cite{BDZ17}. It contains several new contributions, in particular:
\begin{itemize}
\item[Sec. \ref{secdfa}] \begin{itemize} 
\item For DFAs we extend the complete analysis from $n \leq 6$ to $n \leq 7$.
\item We further investigate DFAs with given alphabet size. 
\end{itemize}
\item[Sec. \ref{secpfa}] \begin{itemize} 
\item For PFAs we also extend the analysis to $n \leq 7$, and fine tune it 
by also taking the number of symbols into account.
\item We investigate the carefully synchronizing word lengths for binary PFAs on $n \leq 10$ states.
\end{itemize}
\item[Sec. \ref{secbinpfa}] \begin{itemize} \item We give sequences of binary and ternary PFAs, reaching the maximal possible length for some values of $n$. \end{itemize}
\item[Sec. \ref{sec:asymptotics}] \begin{itemize} \item We improve our asymptotic results and include a 
construction with a single undefined transition. \end{itemize}
\end{itemize}

The most important update to Sections \ref{secexp} and \ref{sectwos} is that transitivity is taken 
into account. Section \ref{secexp} presents our construction of PFAs on three symbols with
exponential shortest carefully synchronizing word length.
In Section \ref{sectwos} we improve the transformation used by Martyugin \cite{M10} and Vorel 
\cite{V16} to reduce to alphabet size two. We conclude in Section \ref{secconcl}.

\section{Critical DFAs on at Most 7 States}
\label{secdfa}

A natural question when studying \v{C}ern\'y's conjecture is: what can be said about automata in which the
bound of the conjecture is actually attained, the so-called critical automata?
Throughout this section we restrict ourselves to basic DFAs. As has already been noted by several authors
\cite{volkov,T06,DZ17}, critical DFAs are rare. There is only one construction known which gives a critical
DFA for each $n$, namely the well-known sequence $C_n$, discovered by and named after \v{C}ern\'y \cite{C64}.
Apart from this sequence, all known critical DFAs have at most 6 states. In \cite{DZ17}, all critical DFAs
on less than 5 states were identified, without restriction on the size of the alphabet. For $n=5$ and $n=6$
it was still an open question if there exist critical (or even supercritical) DFAs,
other than those already discovered by \v{C}ern\'y, Roman \cite{R08} and Kari \cite{K01}. 
In \cite{BDZ17}, we verified that this is not the case, so for $n=5$ only two critical DFAs exist 
(\v{C}ern\'y, Roman) and also for $n=6$ only two exist (\v{C}ern\'y, Kari). Here we extend the analysis 
to $n=7$, for which \v{C}ern\'y's DFA is the only critical DFA. In fact our results also prove the 
following theorem:

\begin{theorem}
	Every synchronizing DFA on $n\leq 7$ states admits a synchronizing word of length at most $(n-1)^2$.
\end{theorem}

As Trahtman already noted in his paper \cite{T06}, for $n\geq 6$ there seems to be a gap in the range of
possible shortest synchronization lengths. For example, his analysis showed that there are no DFAs on 6
states with shortest synchronizing word length 24, and no DFAs on 7 states with length 33, 34 or 35, when 
restricting to at most 4 symbols. Our analysis
shows that this is true without restriction on the alphabet: there is no DFA on 6 states with shortest
synchronizing word length 24. For $n\leq 6$ all other lengths are feasible: if $n\leq 6$ and
$1\leq k\leq (n-1)^2$, $k\neq 24$, then there exists a DFA on $n$ states with shortest synchronizing word
length exactly $k$. For $n=7$ all values $k \leq 32$ occur as shortest synchronizing word length.

As the number of DFAs on $n$ states grows like $2^{n^{n}}$, an exhaustive search is a non-trivial affair,
even for small values of $n$. The problem is that the alphabet size in a basic DFA can be as large as
$n^n-1$. Up to now, for $n=5,6,7$ only DFAs with at most four symbols were checked by Trahtman \cite{T06}.
Here we describe our algorithm to investigate all DFAs on 5, 6 and 7 states, without restriction on
the alphabet size.

Before explaining the algorithm, we introduce some terminology. A DFA $\mathcal{B}$ obtained by adding
some symbols to a DFA $\mathcal{A}$ will be called an \emph{extension} of $\mathcal{A}$. If $\mathcal{A} =
(Q,\Sigma,\delta)$, then $S\subseteq Q$ will be called \emph{reachable} if there exists a word
$w\in\Sigma^*$ such that $Qw=S$. We say
that $S$ is \emph{reducible} if there exists a word $w$ such that $|Sw|<|S|$, and we call $w$ a
\emph{reduction word} for $S$. Our algorithm is mainly based on the following observation:
\begin{property}\label{property:extension}
	If a DFA $\mathcal{A}$ is synchronizing, and $\mathcal{B}$ is an extension of $\mathcal{A}$, then
	$\mathcal {B}$ is synchronizing as well and its shortest synchronizing word is at most as long as the
	shortest synchronzing word for $\mathcal{A}$.
\end{property}

The algorithm roughly runs as follows. We search for (super)critical DFAs on $n$ states, so a DFA is
discarded if it synchronizes faster, or if it does not synchronize at all. For a given DFA
$\mathcal{A} = (Q,\Sigma,\delta)$ which is not yet discarded or investigated, the algorithm does the
following:
\begin{enumerate}
	\item If $\mathcal{A}$ is synchronizing and (super)critical, we have identified an example we are
	searching for.
	\item If $\mathcal{A}$ is synchronizing and subcritical, it is discarded, together with all its
	possible extensions (justified by Property \ref{property:extension}).
	\item If $\mathcal{A}$ is not synchronizing, then find an upper bound $L$ for how fast any synchronizing
	extension of $\mathcal{A}$ will synchronize (see below). If $L < (n-1)^2$, then discard $\mathcal{A}$
	and all its extensions. Otherwise, discard only $\mathcal{A}$ itself.
\end{enumerate}

The upper bound $L$ for how fast any synchronizing extension of $\mathcal{A}$ will synchronize, is
found by analyzing distances in the directed graph of the power automaton of $\mathcal{A}$. For
$S,T\subseteq Q$, the distance $\dist(S,T)$ from $S$ to $T$ in this graph is equal to the 
length of the shortest word $w$ for which $Sw=T$, if such a word exists. 

The distances in the directed graph of the power automaton are computed by way of the 
Floyd-Warshall algorithm. As the computation complexity of the Floyd-Warshall algorithm
is cubic, the complexity in terms of $n$ is $\Theta(8^n)$, which is actually quite bad.
For that reason, we took the effort to implement it far more efficiently
than the straightforward way, see \cite{B17}. 

We do not compute $\dist(S,T)$ if $T$ is a singleton. Instead, we compute
$$
\min\{\dist(S,T) \mid T \subset Q \mbox{ and } |T| \le i \}
$$
for every $S \subseteq Q$ and $i = 1,2,\ldots,n-1$: for $i = 1$ as a replacement,
yielding vacated space in the distance matrix, and for larger $i$ as a usage of 
this space.

A possible upper bound $L$ is as follows:
\begin{enumerate}
	\item Determine the size $|S|$ of a smallest reachable set $S$. 
	Let $m$ be the minimal distance from $Q$ to a set of size $|S|$.
	\item For each $k\leq |S|$, partition the collection of irreducible sets of size $k$ into strongly
	connected components. Let $m_k$ be the number of components plus the sum of their diameters.
	\item For each reducible set $R$ of size $k\leq |S|$, find the length $l_R$ of its shortest
	reduction word. Let $l_k$ be the maximum of these lengths.
	\item Now note that a synchronizing extension of $\mathcal{A}$ will have a synchronizing
	word of length at most
	\[ L \; = \; \sum_{k=2}^{|S|}(m_k+l_k) + m.  \]
\end{enumerate}

A slightly better upper bound is the following. Let $M$ be the maximum distance 
from $Q$ to a set of size $|S|$. Partition the irreducible sets of size $|S|$ 
which can be reached from $Q$ into strongly connected components, and let $c$ 
be the number of components plus the sum of their diameters. Then
a synchronizing extension of $\mathcal{A}$ will have a synchronizing
word of length at most
\[ L' \; = \; \sum_{k=2}^{|S|}(m_k+l_k) - c + 1 + M.  \]
So one can say that $Q$ as a reducible subset is treated differently in 
the construction of $L'$ than in the construction $L$. As a consequence,
$L' \le L$, so $L'$ is a better upper bound than $L$. 
In the upper bound $L''$ which is actually used in the computations, 
we extend this different treatment to other reducible subsets. 

But first, we describe $L$ in an inductive way. We take $L = L_{|S|} + m$,
and define
\begin{align*}
L_1 \; &= \; 0, \\
L_k \; &= \; m_k + l_k + L_{k-1} \\ &= \; m_k + \max\{l_R \mid R 
\mbox{ is reducible and } |R| = k \} + L_{k-1} && \mbox{if $k > 1$.}
\end{align*}
Here, $L_k$ is an upper bound for the maximum
length of the shortest synchronizing word for any subset of size $k$.
We take $L'' = L''_Q$, and we define inductively an upper bound $L''_R$
for the length of the the shortest synchronizing word for 
a reducible subset $R$, and an upper bound $L''_k$ for the maximum
length of the shortest synchronizing word for any subset of size $k$.
Define $S_R$, $m_R$, $M_R$ and $c_R$ as $S$, $m$, $M$ and $c$ respectively, 
but with $Q$ replaced by $R$.
\begin{align*}
L''_R \; &= \; m_R                       && \mbox{if $|S_R| = 1$,} \\
L''_R \; &= \; L''_{|S_R|} - c_R + 1 + M_R && \mbox{if $|S_R| > 1$,} \\
L''_1 \; &= \; 0, \\
L''_k \; &= \; m_k + \max\{ L''_{k-1}, L''_R \mid R 
\mbox{ is reducible and } |R| = k \} && \mbox{if $k > 1$.}
\end{align*}
Although $L''$ yields a better upper bound than $L'$ in general, we do not
always have $L'' \le L'$. To overcome this, we improved the definition
of $L''_R$ in the newest version of the code, but only for $R \ne Q$,
by taking the minimum of what is given above, and $L''_{|R|-1} + l_R$.
(The calculations on DFAs with $7$ states have not been redone.)

The algorithm performs a depth-first search. So after investigating a DFA, first all its extensions (not
yet considered) are investigated before moving on.
Still, we can choose which extension to pick first. We would like to choose
an extension that is likely to be discarded immediately together with all its extensions. Therefore,
we apply the following heuristic: for each possible extension $\mathcal{B}$ by one symbol, we count how
many pairs of states in $\mathcal{B}$ would be reducible. The extension for which this is maximal is
investigated first. The motivation is that a DFA is synchronizing if and only if each pair is reducible 
\cite{C64}.

Furthermore, we only investigate extensions $\mathcal{B}$ by one symbol if either the number of 
pairs which synchronize in $\mathcal{B}$ is larger than in $\mathcal{A}$, or $\mathcal{A}$ 
(and hence also $\mathcal{B}$) is synchronizing. The idea behind this is the following, 
which is easy to prove. If $\mathcal{A}$ is not synchronizing and $\mathcal{B}$ is an extension of 
$\mathcal{A}$ which is synchronizing, then $\mathcal{B}$ has a symbol, which, when added to 
$\mathcal{A}$, increases the number of synchronizing pairs.

The algorithm which has actually been used also takes symmetries on the set of states into account, 
making it almost $n!$ times faster. The symmetry reduction on the states is perfect for automata 
which do not have a pair of conjugate symbols (two symbols $a$ and $b$ are conjugate
if there exists a symmetry $\sigma$ such that $\sigma b \sigma^{-1} = a$). 
Furthermore, we used a multithreaded version of the algorithm for the case of $n = 7$ states. 

In the table below, we counted for every number of symbols (alph.\@ size) and every minimal 
synchronization length (sync.) $\geq 31$, the number of corresponding basic DFAs with
seven states, up to symmetry. We do not require the automata to be minimal, meaning that we allow solutions from which symbols can be removed without changing the synchronization length.  This explains why our numbers differ from those found by Szyku{\l}a in his thesis.
\begin{center}
	\renewcommand{\arraystretch}{1.4}
	\begin{tabular}{|r|rrrrrr|}
		\hline
		alph.\@ & sync.\@ & sync.\@ & sync.\@ & sync.\@ & sync.\@ & sync.\@ \\[-5pt]
		size & 36 & 35 & 34 & 33 & 32 & 31 \\
		\hline
		1 & & & & & & \\[-5pt]
		2 & 1 & & & & 3 & 3 \\[-5pt]
		3 & & & & & 3 & 8 \\[-5pt]
		4 & & & & & & 4 \\
		\hline
		total & 1 & 0 & 0 & 0 & 6 & 15 \\
		\hline
	\end{tabular}
\end{center}
Tables for less than seven states can be found in \cite{BDZ16}, which is an extended version of 
\cite{DZ17}. See also the graph in subsection \ref{secpfaalph}.

In \cite{10.1007/978-3-319-40946-7_15}, the synchronization upper bounds which are used for
pruning the search are different, and only work for DFAs. But there are also differences
in the way the searching is performed. In \cite{10.1007/978-3-319-40946-7_15}, the searching
is done by way of breadth-first search instead of depth-first search, taking far less overhead 
and resources, so it can be done for larger number of states as well. But it leads to
more redundancy in the search: more DFAs need to be scanned, increasing the computation time. It appeared that with our search algorithm, one can compute all critical DFAs with 6 states within a day (using \cite[Theorem 1]{10.1007/978-3-319-40946-7_15} as a synchronization upper bound for pruning; with our synchronization upper bound it is about 400 times faster.)

\section{PFAs with Small State Set}
\label{secpfa}
In the remainder of this paper, we study PFAs and shortest carefully synchronizing word lengths. 
In this section and the next, we focus on PFAs that have long shortest synchronizing words and a 
small number of states. In later sections we construct PFAs on two or three symbols with shortest 
carefully synchronizing words of exponential length for general $n$.

\subsection{PFAs on at Most 7 States}
To find PFAs with small number of states and long shortest carefully synchronizing word, we exploit that Property \ref{property:extension} also holds for PFAs. However, for PFAs it is not true that reducibility of all pairs of states guarantees careful synchronization. Therefore, we apply a different search algorithm.
We search for a PFA with synchronizing length equal to or greater than some given target length. To construct it, we build the alphabet by choosing the symbols of a long shortest synchronizing word from left to right.
More precisely, on the stack of the search function we always have a prefix
of a possible synchronizing word. The search is pruned in the following three cases,
where $w$ is the prefix on the stack:
\begin{enumerate}
	
	\item There exists a word $u$ consisting of the letters of $w$,
	with $|u| < |w|$, such that either $Qu = Qw$, or $Qu$ and $Qw$ are both singletons;
	
	\item The automaton $\mathcal{A}$, whose symbols are the letters
	of $w$, has a synchronizing word which is shorter than the target length;
	
	\item The value of the upper bound $L''$ for the automaton $\mathcal{A}$
	is smaller than the target length.
	
\end{enumerate}
If the search is not pruned, the prefix $w$ will be extended by one letter $a$.
To reduce the number of solutions and speed up the algorithm even further, we only select a candidate symbol $a$ as follows:
\begin{enumerate}
	
	\item If $Qwa = Qwb$ for a letter $b$ of $w$, then $a$ is only selected if it is equal to the first such letter in $w$;
	
	\item If $Qwa = Qwb$ does not hold for any letter $b$ of $w$, then $a$ is only selected if it is undefined outside
	$Qw$.
	
\end{enumerate}
The purpose of symbol $a$ is to get from $Qw$ to the next subset. 
In the situation of (1), no new symbol need to be added to make the transition from $Qw$ to the next
subset. We choose $a$ to be an old symbol, because there is no need to add a new symbol at this point
in the search. In the situation of (2), we choose $a$ to be defined on states
of $Qw$ only, because the purpose of $a$ is to get from $Qw$ to the next subset. There is no need
to add a more complete symbol at this point in the search.

The selection rules (1) and (2) above significantly reduce the number of cases, but (2) has the drawback that 
the algorithm does not necessarily find the solution with the smallest possible alphabet
any more. For example, it did not find a solution of length $37$ with only
$6$ symbols for $n = 6$. But postprocessing all solutions for $n = 6$ did reveal
a solution of length $37$ with only $6$ symbols indeed.

During the postprocessing of a solution, symbols are made more complete, so (2) does not hold
any longer. There are many ways to make the symbols more complete, but most of them will
affect the synchronization length, which gives us effective pruning. For every solution
with more complete symbols, symbols may have become the same, and we count the number of
distinct symbols.

Just as for the DFAs, we took symmetry into account. But we did not need a multithreaded 
version of the algorithm for the case of $n = 7$ states. 

For $n\leq 7$, our algorithm has identified the maximal length $p(n)$ 
of a shortest carefully synchronizing word in a PFA on $n$ states. 
The results are:
\[
\renewcommand{\arraystretch}{1.4}
\begin{array}{|c|c|c|c|c|c|c|}
\hline
\quad n \quad&\quad 2 \quad&\quad 3 \quad&\quad 4 \quad&\quad 5 \quad&\quad 6 \quad&\quad 7 \quad\\
\hline
\quad p(n) \quad& 1 & 4 & 10 & 21 & 37 & 63 \\
\hline
\end{array}
\vspace{3pt}
\]
For $n=8$ states, $102$ can be reached as shortest carefully synchronizing word length, using
$9$ symbols. But $8$ states are too many for us to prove computationally that this is the largest possible length.

Whereas for $n\geq 6$, no critical DFAs are known with more than two symbols, 
PFAs with long shortest carefully synchronizing word lengths tend to have more symbols: for 
$n=4,5,6,7$ states, the minimal numbers of symbols achieving the maximal shortest carefully 
synchronizing word lengths 10, 21, 37 and 63 are 3, 6, 6 and 8 respectively. 
Below we give examples of PFAs on 4, 5, 6 and 7 states reaching these lengths.
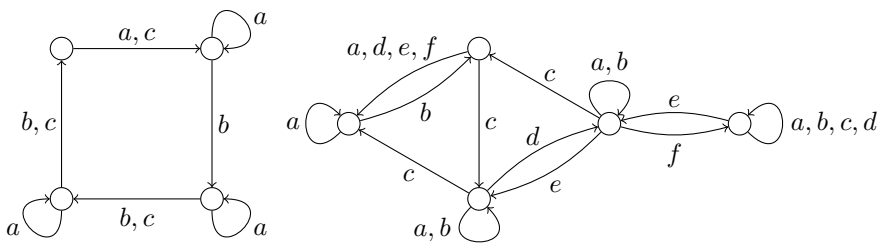
\begin{center}
	\begin{tikzpicture}
	\useasboundingbox (0,-0.5) rectangle (3,2.5);
	\node[circle,draw,inner sep=0pt,minimum width=3mm] (1) at (0,2) {};
	\node[circle,draw,inner sep=0pt,minimum width=3mm] (2) at (2,2) {};
	\node[circle,draw,inner sep=0pt,minimum width=3mm] (3) at (2,0) {};
	\node[circle,draw,inner sep=0pt,minimum width=3mm] (4) at (0,0) {};
	\draw[->] (1) -- node[above,inner sep=1pt] {$a,c$} (2);
	\draw[->] (2) -- node[right,inner sep=2pt] {$b$} (3);
	\draw[->] (3) -- node[below,inner sep=3pt] {$b,c$} (4);
	\draw[->] (4) -- node[left,inner sep=2pt] {$b,c$} (1);
	\draw[->] (2) edge[out=90,in=0,looseness=10] node[right,inner sep=4pt] {$a$} (2);
	\draw[->] (3) edge[out=-90,in=0,looseness=10] node[right,inner sep=4pt] {$a$} (3);
	\draw[->] (4) edge[out=-90,in=-180,looseness=10] node[left,inner sep=4pt] {$a$} (4);
	\end{tikzpicture}
	\begin{tikzpicture}
	\useasboundingbox (-0.7,-0.5) rectangle (6.2,2.5);
	\node[circle,draw,inner sep=0pt,minimum width=3mm] (0) at (0,1) {};
	\node[circle,draw,inner sep=0pt,minimum width=3mm] (1) at (1.732,2) {};
	\node[circle,draw,inner sep=0pt,minimum width=3mm] (2) at (1.732,0) {};
	\node[circle,draw,inner sep=0pt,minimum width=3mm] (3) at (3.464,1) {};
	\node[circle,draw,inner sep=0pt,minimum width=3mm] (4) at (5.196,1) {};
	\draw[->] (0) edge[out=15,in=-135] node[pos=0.55,below,inner sep=3pt] {$b$} (1);
	\draw[->] (1) edge[out=-165,in=45] node[pos=0.65,above,inner sep=9pt] {$a,d,e,f$} (0);
	\draw[->] (1) -- node[right,inner sep=2pt] {$c$} (2);
	\draw[->] (2) -- node[pos=0.55,below,inner sep=4pt] {$c$} (0);
	\draw[->] (3) -- node[pos=0.45,above,inner sep=3pt] {$c$} (1);
	\draw[->] (2) edge[out=45,in=-165] node[pos=0.45,above,inner sep=3pt] {$d$} (3);
	\draw[->] (3) edge[out=-135,in=15] node[pos=0.45,below,inner sep=4pt] {$e$} (2);
	\draw[->] (4) edge[out=165,in=15] node[above,inner sep=2pt] {$e$} (3);
	\draw[->] (3) edge[out=-15,in=-165] node[below,inner sep=3pt] {$f$} (4);
	\draw[->] (0) edge[out=225,in=135,looseness=10] node[left,inner sep=2pt] {$a$} (0);
	\draw[->] (2) edge[out=225,in=315,looseness=10] node[pos=0.2,left,inner sep=3pt] {$a,b$} (2);
	\draw[->] (3) edge[out=135,in=45,looseness=10] node[above,inner sep=2pt] {$a,b$} (3);
	\draw[->] (4) edge[out=315,in=45,looseness=10] node[right,inner sep=3pt] {$a,b,c,d$} (4);
	\end{tikzpicture}
\end{center}
The left one has two synchronizing words of length 10:
$abcabab\bcdot(b+c)ca$. The right one has unique shortest synchronizing word 
$\mathit{abcabdbebcabdbfbcdeca}$ of length 21.
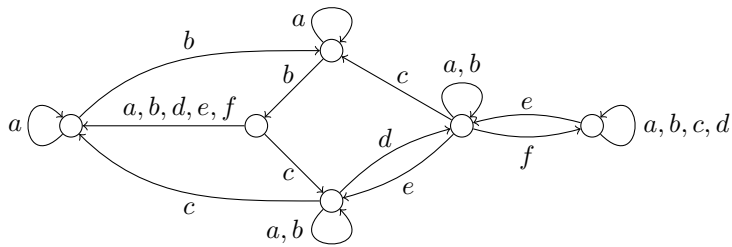
\begin{center}
	\begin{tikzpicture}
	\useasboundingbox (-0.7,-0.5) rectangle (8.5,2.5);
	\node[circle,draw,inner sep=0pt,minimum width=3mm] (0) at (0,1) {};
	\node[circle,draw,inner sep=0pt,minimum width=3mm] (1) at (2.464,1) {};
	\node[circle,draw,inner sep=0pt,minimum width=3mm] (2) at (3.464,2) {};
	\node[circle,draw,inner sep=0pt,minimum width=3mm] (3) at (3.464,0) {};
	\node[circle,draw,inner sep=0pt,minimum width=3mm] (4) at (5.196,1) {};
	\node[circle,draw,inner sep=0pt,minimum width=3mm] (5) at (6.928,1) {};
	\draw[->] (0) edge[out=45,in=-180] node[above]{$b$} (2);
	\draw[->] (1) -- node[pos=0.40,above,inner sep=2pt]{$a,b,d,e,f$} (0);
	\draw[->] (1) -- node[pos=0.40, below,inner sep=4pt]{$c$} (3);
	\draw[->] (2) -- node[pos=0.60,above,inner sep=4pt]{$b$} (1);
	\draw[->] (3) edge[out=-180,in=-45] node[below]{$c$} (0);
	\draw[->] (3) edge[out=45,in=-165] node[pos=0.45,above,inner sep=3pt] {$d$} (4);
	\draw[->] (4) -- node[pos=0.45,above,inner sep=3pt]{$c$}(2);
	\draw[->] (4) edge[out=-135,in=15] node[pos=0.45,below,inner sep=4pt] {$e$} (3);
	\draw[->] (4) edge[out=-15,in=-165] node[below,inner sep=3pt] {$f$} (5);
	\draw[->] (5) edge[out=165,in=15] node[above,inner sep=2pt] {$e$}(4);
	\draw[->] (0) edge[out=225,in=135,looseness=10] node[left,inner sep=2pt] {$a$} (0);
	\draw[->] (2) edge[out=135,in=45,looseness=10] node[pos=0.2,left,inner sep=2pt] {$a$} (2);
	\draw[->] (3) edge[out=225,in=315,looseness=10] node[pos=0.2,left,inner sep=3pt] {$a,b$} (3);
	\draw[->] (4) edge[out=135,in=45,looseness=10] node[above,inner sep=2pt] {$a,b$} (4);
	\draw[->] (5) edge[out=315,in=45,looseness=10] node[right,inner sep=3pt] {$a,b,c,d$} (5);
	\end{tikzpicture}
\end{center}
The shortest synchronizing word is $ab^2ab^2cb^2ab^2db^2eb^2cb^2ab^2db^2fb^2cdecb^2a$ for this 
PFA on 6 states. It is unique and has length 37.
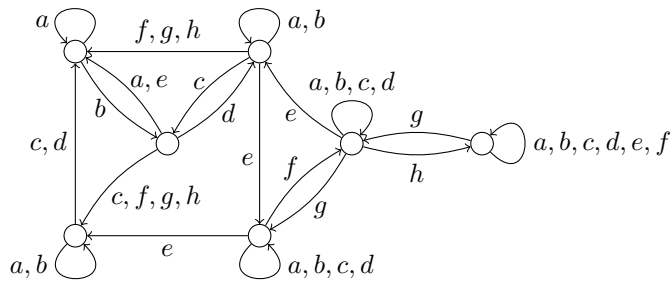
\begin{center}
	\begin{tikzpicture}[x=1.225cm,y=1.225cm]
	\useasboundingbox (-0.6,-1.4) rectangle (6.2,1.4);
	\node[circle,draw,inner sep=0pt,minimum width=3mm] (0) at (0,1) {};
	\node[circle,draw,inner sep=0pt,minimum width=3mm] (1) at (1,0) {};
	\node[circle,draw,inner sep=0pt,minimum width=3mm] (2) at (0,-1) {};
	\node[circle,draw,inner sep=0pt,minimum width=3mm] (3) at (2,1) {};
	\node[circle,draw,inner sep=0pt,minimum width=3mm] (4) at (2,-1) {};
	\node[circle,draw,inner sep=0pt,minimum width=3mm] (5) at (3,0) {};
	\node[circle,draw,inner sep=0pt,minimum width=3mm] (6) at (4.414,0) {};
	\draw[->] (3) -- node[above,inner sep=3pt]{$f,g,h$} (0);
	\draw[->] (1) edge[in=-30,out=120] node[pos=0.2,above,inner sep=8pt]{$a,e$} (0);
	\draw[->] (0) edge[out=-60,in=150] node[pos=0.3,below,inner sep=3pt]{$b$} (1);
	\draw[->] (2) -- node[left,inner sep=3pt]{$c,d$} (0);
	\draw[->] (1) edge[out=-150,in=60] node[pos=0.7,right,inner sep=4pt]{$c,f,g,h$} (2);
	\draw[->] (1) edge[out=30,in=-120] node[pos=0.6,below,inner sep=2pt]{$d$} (3);
	\draw[->] (3) edge[in=60,out=-150] node[pos=0.6,above,inner sep=3pt]{$c$} (1);
	\draw[->] (3) -- node[pos=0.6,left,inner sep=2pt]{$e$} (4);
	\draw[->] (4) -- node[below]{$e$} (2);
	\draw[->] (4) edge[out=60,in=-150] node[pos=0.4,above,inner sep=3pt] {$f$} (5);
	\draw[->] (5) edge[out=150,in=-60] node[pos=0.6,below,inner sep=4pt]{$e$} (3);
	\draw[->] (5) edge[out=-120,in=30] node[pos=0.4,below,inner sep=4pt] {$g$} (4);
	\draw[->] (5) edge[out=-15,in=-165] node[below,inner sep=3pt] {$h$} (6);
	\draw[->] (6) edge[out=165,in=15] node[above,inner sep=2pt] {$g$} (5);
	\draw[->] (0) edge[out=45,in=135,looseness=10] node[pos=0.8,left,inner sep=2pt] {$a$} (0);
	\draw[->] (3) edge[out=45,in=135,looseness=10] node[pos=0.2,right,inner sep=3pt] {$a,b$} (3);
	\draw[->] (2) edge[out=225,in=315,looseness=10] node[pos=0.2,left,inner sep=3pt] {$a,b$} (2);
	\draw[->] (4) edge[out=225,in=315,looseness=10] node[pos=0.8,right,inner sep=3pt] {$a,b,c,d$} (4);
	\draw[->] (5) edge[out=135,in=45,looseness=10] node[above,inner sep=2pt] {$a,b,c,d$} (5);
	\draw[->] (6) edge[out=315,in=45,looseness=10] node[right,inner sep=3pt] {$a,b,c,d,e,f$} (6);
	\end{tikzpicture}
\end{center}
There are 81 shortest synchronizing words (of length 63) for this PFA on 7 states, all being of the form
$$
abcabdbebcabdbfbdbgbdbebcabdbfbdbhbdbeb................bdefgeca.
$$
This word is remarkably similar to the one for 5 states and also the actions of some of the symbols are comparable. It is however not yet sufficient to detect a pattern that could be extrapolated to larger $n$. 

\subsection{PFAs on at Most 7 States with Fixed Alphabet Size} 
\label{secpfaalph}

Write $p(n,k)$ for the maximal shortest carefully synchronizing word length for a PFA on $n$ states and $k$ symbols. Computing the values of $p(n,k)$ for all $n \le 7$ and all $k \le 41$ is a lot more involved than computing $p(n)$ for all $n \le 7$. We made several improvements to the algorithm to get it done, among which the following:
\begin{enumerate}
	
	\item It appeared that most of the times where upper bound $L''$ needs to be determined, the 
	PFA is already synchronizing. So we start with trying a breadth first search with bit vectors, 
	and only compute $L''$ in the above-described way if the PFA is not synchronizing.
	
	\item We estimate the number of required symbols after postprocessing (making symbols
	more complete) already before the postprocessing, and use this estimate to prune the search.
	
	\item If the estimate on the number of required symbols is equal to the maximum allowed 
	number of symbols, then for every extension $\mathcal{B}$ of $\mathcal{A}$, the PFAs we get
	by postprocessing $\mathcal{B}$ are contained in the PFAs we get by postprocessing $\mathcal{A}$
	directly. For that reason, we do not search further for extensions of $\mathcal{A}$ in this case, 
	but postprocess immediately. So the postprocessing is not only to reduce the number of 
	symbols in this case, but also to obtain synchronization.
	
\end{enumerate}
In the graph below, the values of $p(n,k)$ are plotted for all $n \le 7$ and all 
$k \le 40$ in light gray. Furthermore, the values of $d(n,k)$ for DFAs are plotted for all $n \le 7$ and all $k \le 40$ in dark gray, except the cases where $n = 7$ and $5 \le k \le 40$.

\begin{figure}[t]
	\begin{center}
		\begin{tikzpicture}[x=2mm,y=2mm]
		\begin{scope}
		\clip (0,0) rectangle (40.7,63.7);
		\foreach \x in {0,5,...,40} {
			\foreach \y in {0,10,...,60} {
				\fill[black!10] (\x,{\y+5*mod(\x,2)}) rectangle ++(5,5);
			}
		}
		\end{scope}
		\foreach \x in {5,10,...,40} {
			\draw (\x,0) node[black,anchor=north] {$\scriptstyle\x$};
		}
		\foreach \y in {5,10,...,60} {
			\draw (0,\y) node[black,anchor=east] {$\scriptstyle\y\!$};
		}
		\draw (0,63.7) -- (0,-0.7) (40.7,0) -- (-0.7,0);
		\mathversion{bold}
		\draw[black!30,thick,line cap=round,line join=round]
		(1,1) circle (1pt) 
		\foreach \s in {2,3} { -- (\s,1) circle (1pt) }
		(1,2) circle (1pt) 
		\foreach \s in {2,...,16} { -- (\s,4) circle (1pt) }
		\foreach \s in {17,...,20} { -- (\s,3) circle (1pt) }
		\foreach \s in {21,...,23} { -- (\s,2) circle (1pt) }
		\foreach \s in {24,...,26} { -- (\s,1) circle (1pt) }
		(1,3) circle (1pt)
		\foreach \s in {2} { -- (\s,9) circle (1pt) }
		\foreach \s in {3,...,40} { -- (\s,10) circle (1pt) }
		edge[line cap=butt] (40.7,10) node[anchor=west,shift={(0,-1)},scale=1.25] {\bf$n=4$ PFA}
		(1,4) circle (1pt)
		\foreach \s in {2} { -- (\s,16) circle (1pt) }
		\foreach \s in {3} { -- (\s,19) circle (1pt) }
		\foreach \s in {4,5} { -- (\s,20) circle (1pt) }
		\foreach \s in {6,...,40} { -- (\s,21) circle (1pt) }
		edge[line cap=butt] (40.7,21) node[anchor=west,scale=1.25] {\bf$n=5$ PFA}
		(1,5) circle (1pt)
		\foreach \s in {2} { -- (\s,26) circle (1pt) }
		\foreach \s in {3} { -- (\s,33) circle (1pt) }
		\foreach \s in {4} { -- (\s,34) circle (1pt) }
		\foreach \s in {5} { -- (\s,36) circle (1pt) }
		\foreach \s in {6,...,40} { -- (\s,37) circle (1pt) }
		edge[line cap=butt] (40.7,37) node[anchor=west,scale=1.25] {\bf$n=6$ PFA}
		(1,6) circle (1pt)
		\foreach \s in {2} { -- (\s,39) circle (1pt) }
		\foreach \s in {3} { -- (\s,51) circle (1pt) }
		\foreach \s in {4} { -- (\s,54) circle (1pt) }
		\foreach \s in {5} { -- (\s,58) circle (1pt) }
		\foreach \s in {6} { -- (\s,60) circle (1pt) }
		\foreach \s in {7} { -- (\s,62) circle (1pt) }
		\foreach \s in {8,...,40} { -- (\s,63) circle (1pt) }
		edge[line cap=butt] (40.7,63) node[anchor=west,scale=1.25] {\bf$n=7$ PFA};
		\mathversion{bold}
		\draw[black!50,thick,line cap=round,line join=round]
		(1,1) circle (1pt)
		\foreach \s in {2,3} { -- (\s,1) circle (1pt) }
		node[anchor=west,inner sep=2pt,scale=1.25] {$n=2$}
		(1,2) circle (1pt) 
		\foreach \s in {2,...,5} { -- (\s,4) circle (1pt) }
		\foreach \s in {6,...,9} { -- (\s,3) circle (1pt) }
		\foreach \s in {10,...,23} { -- (\s,2) circle (1pt) }
		\foreach \s in {24,...,26} { -- (\s,1) circle (1pt) }
		node[anchor=west,inner sep=2pt,scale=1.25] {$n=3$}
		(1,3) circle (1pt) 
		\foreach \s in {2,...,5} { -- (\s,9) circle (1pt) }
		\foreach \s in {6,...,8} { -- (\s,8) circle (1pt) }
		\foreach \s in {9,...,17} { -- (\s,7) circle (1pt) }
		\foreach \s in {18,...,40} { -- (\s,5) circle (1pt) }
		edge[line cap=butt] (40.7,5) node[anchor=west,scale=1.25] {\bf$n=4$ DFA}
		(1,4) circle (1pt)
		\foreach \s in {2,3} { -- (\s,16) circle (1pt) }
		\foreach \s in {4,...,6} { -- (\s,15) circle (1pt) }
		\foreach \s in {7,...,13} { -- (\s,14) circle (1pt) }
		\foreach \s in {14,15} { -- (\s,13) circle (1pt) }
		\foreach \s in {16,...,23} { -- (\s,12) circle (1pt) }
		\foreach \s in {24,...,29} { -- (\s,11) circle (1pt) }
		\foreach \s in {30,...,40} { -- (\s,10) circle (1pt) }
		edge[line cap=butt] (40.7,10) node[anchor=west,shift={(0,1)},scale=1.25] {\bf$n=5$ DFA}
		(1,5) circle (1pt)
		\foreach \s in {2} { -- (\s,25) circle (1pt) }
		\foreach \s in {3} { -- (\s,23) circle (1pt) }
		\foreach \s in {4,...,11} { -- (\s,22) circle (1pt) }
		\foreach \s in {12,...,15} { -- (\s,21) circle (1pt) }
		\foreach \s in {16,...,21} { -- (\s,20) circle (1pt) }
		\foreach \s in {22,...,40} { -- (\s,19) circle (1pt) }
		edge[line cap=butt] (40.7,19) node[anchor=west,scale=1.25] {\bf$n=6$ DFA}
		(1,6) circle (1pt)
		\foreach \s in {2} { -- (\s,36) circle (1pt) }
		\foreach \s in {3} { -- (\s,32) circle (1pt) }
		\foreach \s in {4} { -- (\s,31) circle (1pt) }
		\foreach \s in {5} { -- (\s,30) circle (1pt) }
		node[anchor=west,inner sep=2pt,scale=1.25] {\bf$n=7$ DFA, $k \le 4$};
		\mathversion{normal}
		\draw (17.5,-2) node {$k$} edge[->] (23,-2)
		(-4,28.4) node {$d(n,k)$} edge[->] (-4,37)
		(-4,26.4) node {$p(n,k)$};
		\end{tikzpicture}
	\end{center}
\end{figure}
So, we see that for DFAs with $n \le 7$ states, after having the maximum 
$d(n,k) = (n-1)^2$ at $k = 2$, the values of $d(n,k)$ decrease for larger $k$. 
So it seems that for DFAs with a greater number of symbols, it is harder to get 
large synchronization lengths.

For PFAs, this behaviour is quite different. Due to partiality, symbols
may be only applicable on a few subsets of the set of all states, which gives less
possibilities to synchronize carefully and therefore more possibilities for coexistence of symbols in a slowly synchronizing PFA.

\subsection{Binary DFAs and PFAs on at Most 10 States}

Now that we know that the maximal carefully synchronization lengths of PFAs with $n$ states 
are larger than the synchronization lengths of DFAs, we can wonder what will
happen if we fix the alphabet size to $2$. For DFAs all evidence suggests that this choice gives the largest possible synchronization lengths. In contrast, for binary PFAs the lengths grow slower than for general PFAs, although the growth is still exponential as we will see in Section \ref{sectwos}.  

Using breadth first search with bit vectors, combined with symmetry reduction on the states, 
we computed all possible carefully synchronization lengths of binary PFAs with $n \le 9$ states. 
For the binary PFAs with $n = 10$ states, we additionally
used multithreading and applied a few low level optimizations. 
One of the optimization techniques was to view the PFAs as CNFAs, 
namely by replacing undefined transitions by transitions to the 
whole set of states.
The results are displayed below, where the maximum carefully 
synchronization lengths $p(n,2)$ are in boldface. For comparison, we also added the known 
synchronization lengths for general PFAs.
\[
\renewcommand{\arraystretch}{1.4}
\begin{array}{|c|l|l|l|}
\hline
\quad n \quad&\quad \mbox{binary DFA} \quad&\quad \mbox{proper binary PFA} \quad &\quad \mbox{PFA}\quad\\
\hline
\quad 2 \quad&\quad {\bf 1} 
\quad&\quad {\bf 1} \quad &\quad 1\\
\hline
\quad 3 \quad&\quad 1\mbox{--}{\bf 4} 
\quad&\quad 1\mbox{--}3 \quad &\quad 1\mbox{--}4\\
\hline
\quad 4 \quad&\quad 1\mbox{--}{\bf 9} 
\quad&\quad 1\mbox{--}7 \quad &\quad 1\mbox{--}10\\
\hline
\quad 5 \quad&\quad 1\mbox{--}{\bf 16} 
\quad&\quad 1\mbox{--}15 \quad &\quad 1\mbox{--}21\\
\hline
\quad 6 \quad&\quad 1\mbox{--}23,\quad 25 
\quad&\quad 1\mbox{--}23,\quad {\bf 26} \quad &\quad 1\mbox{--}37\\
\hline
\quad 7 \quad&\quad 1\mbox{--}32,\quad 36 
\quad&\quad 1\mbox{--}33,\quad 35\mbox{--}36,\quad {\bf 39} \quad &\quad 1\mbox{--}63\\
\hline
\quad 8 \quad&\quad 1\mbox{--}44,\quad 49 
\quad&\quad 1\mbox{--}45,\quad 48,\quad 50,\quad 52,\quad {\bf 55} \quad &\\
\hline
\quad 9 \quad&\quad 1\mbox{--}52,\quad 56\mbox{--}58,\quad 64 
\quad&\quad 1\mbox{--}63,\quad 65,\quad 68,\quad 72\mbox{--}{\bf 73} \quad &\\
\hline
\quad 10 \quad&\quad 1\mbox{--}66,\quad 72\mbox{--}74,\quad 81 
\quad&\quad 1\mbox{--}80,\quad 82\mbox{--}84,\quad 87,\quad 89,\quad 93\mbox{--}{\bf 94} \quad &\\
\hline
\end{array}
\vspace{3pt}
\]
A notable feature in this table is that several gaps appear in the ranges of possible values. 
Unfortunately, we still lack a deeper understanding of this behaviour. 
For DFAs, existence of gaps has already been observed in \cite{T06,AVG12,10.1007/978-3-319-40946-7_15} and has been studied further in \cite{dzyga_et_al:LIPIcs:2017:8122}.

\section{Specific PFA Constructions}
\label{secbinpfa}

In this section we present two series of PFAs (parameterized by its size $n$) of special interest: 
they have quadratic shortest synchronizing word length exceeding $(n-1)^2$, for each $n$ for which
this is possible. Furthermore, they reach the maximum possible synchronization length for some low
values of $n$. The constructed series fill up a void between the computations up to $7$ or $10$
states respectively, and the asymptotic results in the next sections.

Both series are closely related to \v{C}ern\'y's DFAs. The 
first series is $T_n$ on $n$ states and three symbols; for this we give the full analysis which 
is quite straightforward. The second series is $P_n$ on $n$ states and two symbols. 
For this series the full analysis is much more involved; in this paper we give the construction and 
the results, but the full analysis leading to these results will be presented in a separate paper.

We start by $T_n$. For $n\geq 4$, $T_n$ is defined to be the PFA on the $n$ states $1,2,\ldots,n$ and the three symbols $a,b,c$ such that
\begin{gather*}
qa = \left\{ \begin{array}{ll} 
q+1\quad & 1 \le q \le n-2 \\
1        & q = n-1 \\
n        & q = n
\end{array} \right. \qquad
qc = \left\{ \begin{array}{ll} 
2 & q = 1 \\
\bot      & 2 \le q \le n-1 \\
2 + \lfloor\frac{n-1}{2} \rfloor\quad    & q = n
\end{array} \right.\\
qb = \left\{ \begin{array}{ll} 
2\quad & q = 1 \\
q      & 2 \le q \le n
\end{array} \right. 
\end{gather*}
For $n=3$ we take the same definition in which $nc = 2 + \lfloor\frac{n-1}{2} \rfloor$ is taken modulo $n-1$, so $3c = 1$.
Note that for all $n$ the PFA is obtained by extending $C_{n-1}$ by an extra node $n$ on which $a$ and $b$ act as the identity, and an extra symbol $c$ that is only defined on 1 and $n$. The PFA $T_n$ under consideration is depicted below for $n=7$.
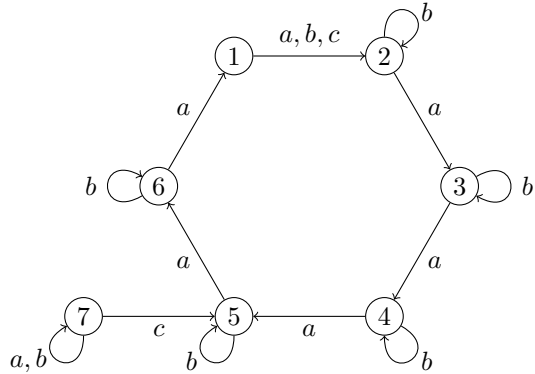
\begin{center}
	\begin{tikzpicture}
	\node[circle,draw,inner sep=0pt,minimum width=5mm] (1) {$1$};
	\node[circle,draw,inner sep=0pt,minimum width=5mm] (2) at ($ (1) +  (0:2) $) {$2$};
	\node[circle,draw,inner sep=0pt,minimum width=5mm] (3) at ($ (2) +  (300:2)$) {$3$};
	\node[circle,draw,inner sep=0pt,minimum width=5mm] (4) at ($ (3) + (240:2)$) {$4$};
	\node[circle,draw,inner sep=0pt,minimum width=5mm] (5) at ($ (4) + (180:2)$) {$5$};
	\node[circle,draw,inner sep=0pt,minimum width=5mm] (6) at ($ (5) + (120:2)$) {$6$};
	\node[circle,draw,inner sep=0pt,minimum width=5mm] (7) at ($ (5) + (180:2)$) {$7$};
	\draw[->] (1) -- node[above,inner sep=3pt] {$a,b,c$} (2);
	\draw[->] (2) -- node[above right,inner sep=2pt] {$a$} (3);
	\draw[->] (3) -- node[below right,inner sep=2pt] {$a$} (4);
	\draw[->] (4) -- node[below,inner sep=3pt] {$a$} (5);
	\draw[->] (7) -- node[below,inner sep=3pt] {$c$} (5);
	\draw[->] (5) -- node[below left,inner sep=2pt] {$a$} (6);
	\draw[->] (6) -- node[above left,inner sep=2pt] {$a$} (1);
	\draw[->] (2) edge[out=90,in=30,looseness=7] node[right,inner sep=4pt] {$b$} (2);
	\draw[->] (3) edge[out=30,in=-30,looseness=7] node[right,inner sep=4pt] {$b$} (3);
	\draw[->] (4) edge[out=-30,in=-90,looseness=7] node[right,inner sep=4pt] {$b$} (4);
	\draw[->] (5) edge[out=-90,in=-150,looseness=7] node[left,inner sep=4pt] {$b$} (5);
	\draw[->] (6) edge[out=-150,in=-210,looseness=7] node[left,inner sep=4pt] {$b$} (6);
	\draw[->] (7) edge[out=-90,in=-150,looseness=7] node[left,inner sep=4pt] {$a,b$} (7);
	\end{tikzpicture}
\end{center}

\begin{theorem}
	For every $n\geq 3$ the PFA $T_n$ is carefully synchronizing with unique shortest synchronizing word $(b a^{n-2})^{n-2}cv$ of length $\frac{3(n-1)(n-2)}{2} + 1$, where
	$v = (a^{n-2}b)^{(n-2)/2}$ if $n$ is even, and $v = a^{(n-3)/2} b (a^{n-2}b)^{(n-3)/2}$ if $n$ is odd.
\end{theorem}
\begin{proof}
	First we show that the given word is carefully synchronizing. Write $Q = \{1,\ldots,n\}$. Since $C_{n-1}$ synchronizes with $(b a^{n-2})^{n-3}b$ ending in state 2, we obtain
	$Q (b a^{n-2})^{n-3}b = \{2,n\}$, followed by $a^{n-2}$ yielding $Q (b a^{n-2})^{n-2} = \{1,n\}$, being the set on which $c$ is defined, hence
	$Q (b a^{n-2})^{n-2} c = \{2,2 + \lfloor\frac{n-1}{2} \rfloor\}$. It is easily checked that in $C_{n-1}$ one has $\{2,2 + \lfloor\frac{n-1}{2} \rfloor\}v = \{2\}$, passing all $\binom{n-1}{2}$ subsets of size 2 of $\{1,\ldots,n-1\}$ exactly once, and decreasing the distance between the two elements by 1 every time a $b$ from $v$ is processed.
	
	Conversely, let $w$ be a shortest carefully synchronizing word for $T_n$. 
	To include the state $n$ in synchronization, $w$ should contain a $c$, 
	so write $w = w_1 c w_2$ in which $w_1 \in \{a,b\}^*$. Since $c$ should be 
	defined on $Q w_1$, we have $Q w_1 \subseteq \{1,n\}$. Ignoring state $n$, 
	we obtain $\{1,\ldots,n-1\} w_1 = \{1\}$ in $C_{n-1}$. Since the shortest 
	prefix of $w_1$ that synchronizes in $C_{n-1}$, synchronizes in state 2, 
	$n-2$ more $a$ steps are needed to synchronize in state 1, so $w_1$ has 
	length at least $n-2$ plus the shortest synchronization length of $C_{n-1}$ 
	being $(n-2)^2$, yielding $|w_1| \geq (n-1)(n-2)$. Note that the 
	synchronizing word we gave satisfies $|w_1| = (n-1)(n-2)$. Since both 1 and 
	$n$ are contained in $Q w_1$, we obtain 
	$Q w_1 c = \{2,2 + \lfloor\frac{n-1}{2} \rfloor\}$. 
	
	So until the singleton is 
	obtained after applying $w_2$ to this set, all intermediate sets consist of two 
	elements from $\{1,\ldots,n-1\}$. One checks that the distance between these two 
	elements can only decrease by a $b$ step, and only in the case the set contains 
	state 1 and a state in $\{2,3,\ldots,\lfloor \frac{n-1}2 \rfloor\}$. Synchronization 
	is obtained if this distance becomes 0.
	Counting the numbers of $b$s and the numbers of intermediate $a$ steps required to 
	satisfy this requirement shows that $v$ is the shortest candidate for $w_2$. Hence no 
	shorter carefully synchronizing word is possible than the one we gave.
\end{proof}

Note that for $n=3,4$ the PFA $T_n$ has the highest possible carefully synchronizing word length 
among all PFAs (4 and 10), while for $n=5$ it is the highest possible among all PFAs on 3 symbols. 
Moreover, for all $n \geq 4$ it strictly exceeds $(n-1)^2$. Furthermore, one can adapt symbol $c$
of $T_n$, to obtain a PFA of which the synchronization length is any given number in 
$\{1,2,\ldots,\frac{3(n-1)(n-2)}{2}\}$.

A natural question is what are the worst cases for binary PFAs and small $n$. 
It turns out that we can find a similar class of binary PFAs, so with two symbols rather than three, 
and having similar properties. 
In the class of binary PFAs, \v{C}ern\'y's example still is the worst possible for $n\leq 5$. 
For $6\leq n\leq 10$, there is a unique binary PFA reaching the maximal length, 
being 26, 39, 55, 73 and 94 for $n = 6,7,8,9,10$ respectively. 
The first four of these PFAs are all members of a sequence $P_n$ that we introduce now. 
Again it looks very much like \v{C}ern\'y's sequence. For $n\geq 3$, $P_n$ is defined by
\[
qa = \left\{ \begin{array}{ll} 
\bot     & q = 1 \\
q+1\quad & 2 \le q \le n-1 \\
1        & q = n
\end{array} \right. \qquad
qb = \left\{ \begin{array}{ll} 
q+1\quad & 1 \le q \le 2 \\
q        & 3 \le q \le n
\end{array} \right.
\]
The PFA $P_n$ under consideration is depicted below for $n=6$.
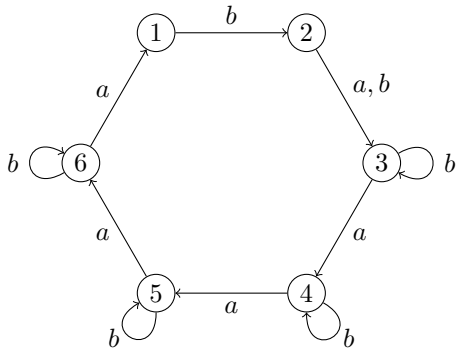
\begin{center}
	\begin{tikzpicture}
	\node[circle,draw,inner sep=0pt,minimum width=5mm] (1) {$1$};
	\node[circle,draw,inner sep=0pt,minimum width=5mm] (2) at ($ (1) +  (0:2) $) {$2$};
	\node[circle,draw,inner sep=0pt,minimum width=5mm] (3) at ($ (2) +  (300:2)$) {$3$};
	\node[circle,draw,inner sep=0pt,minimum width=5mm] (4) at ($ (3) + (240:2)$) {$4$};
	\node[circle,draw,inner sep=0pt,minimum width=5mm] (5) at ($ (4) + (180:2)$) {$5$};
	\node[circle,draw,inner sep=0pt,minimum width=5mm] (6) at ($ (5) + (120:2)$) {$6$};
	\draw[->] (1) -- node[above,inner sep=3pt] {$b$} (2);
	\draw[->] (2) -- node[pos=0.55,above right,inner sep=2pt] {$a,b$} (3);
	\draw[->] (3) -- node[pos=0.45,below right,inner sep=2pt] {$a$} (4);
	\draw[->] (4) -- node[below,inner sep=3pt] {$a$} (5);
	\draw[->] (5) -- node[pos=0.55,below left,inner sep=2pt] {$a$} (6);
	\draw[->] (6) -- node[pos=0.45,above left,inner sep=2pt] {$a$} (1);
	\draw[->] (3) edge[out=30,in=-30,looseness=7] node[right,inner sep=4pt] {$b$} (3);
	\draw[->] (4) edge[out=-30,in=-90,looseness=7] node[right,inner sep=4pt] {$b$} (4);
	\draw[->] (5) edge[out=-90,in=-150,looseness=7] node[left,inner sep=4pt] {$b$} (5);
	\draw[->] (6) edge[out=-150,in=-210,looseness=7] node[left,inner sep=4pt] {$b$} (6);
	\end{tikzpicture}
\end{center}

Although the construction of $P_n$ is quite simple, the synchronization lengths show a somewhat curious pattern. Where sequences of DFAs in the literature generally give rise to quadratic or linear formula's, this is not the case for $P_n$. The lengths are quadratic in size, but no explicit quadratic formula for it exists. The synchronization length of $P_n$ is strictly larger than $(n-1)^2$ for all $n\geq 6$. As $b$ and $ab$ act as  $C_{n-1}$ on $\{2,\ldots,n\}$ it is easily seen that $b(b (ab)^{n-2})^{n-3} b$ is a carefully synchronizing word for $P_n$, but for $n \geq 5$ it is not a shortest one.

The synchronization length can be expressed in the Fibonacci numbers $\fib{(m)}$ defined by 
$\fib{(0)}=0$, $\fib{(1)} = 1$ and $\fib{(m)} = \fib{(m-1)}+\fib{(m-2)}$ for $m\geq 2$. 
Let $\phi = \frac{1+\sqrt{5}}{2}$ be the golden ratio.

\begin{theorem}\label{th:Pn}
	For $n\geq 3$, let $m$ be the unique integer for which $\fib{(m-1)}<n-2\leq \fib{(m)}$. 
	If $w$ is a shortest synchronizing word for $P_n$, then
	\[
	|w| = n^2+mn-5n-\fib{(m+1)}-2m+8 = n^2+\frac{n\log(n)}{\log(\phi)}+\Theta(n).
	\]
	Furthermore, $|w|>(n-1)^2$ for $n\geq 6$.
\end{theorem}

The proof of this theorem is expected to appear in a forthcoming paper by
Stijn Cambie and the first two authors. Below is a table of $|w|$ for small $n$.
\[
\renewcommand{\arraystretch}{1.4}
\begin{array}{|c|c|c|c|c|c|c|c|c|c|}
\hline
\quad n\quad &\quad 3\quad &\quad 4\quad &\quad 5\quad &\quad 6\quad &\quad 7
\quad &\quad 8\quad &\quad 9\quad &\quad 10\quad &\quad 11\quad\\
\hline
|w| & 2 & 7 & 15 & 26 & 39 & 55 & 73 & 93 & 116 \\
\hline
\end{array}
\vspace{3pt}
\]

\section{Exponential Bounds for PFAs on Three Symbols}
\label{secexp}
In this section, we demonstrate our techniques to construct PFAs with only three symbols and exponential shortest synchronizing word length. These constructions are based on string rewrite systems. In the next section we will show a reduction to two symbols and the last section is devoted to more elaborate constructions that lead to sharper asymptotic results. 

For any $k \geq 3$, we build a transitive PFA on $n = 3k$ states and three symbols, which is carefully synchronizing, and the shortest carefully synchronizing word has length $\Omega(\phi^{n/3})$, where $\phi = \frac{1+ \sqrt{5}}{2} \approx 1.618$. The set of states is
$Q = \{A_i, B_i, C_i \mid i = 1,\ldots,k\}$.
If a set $S \subseteq Q$ contains exactly one element of $\{A_i, B_i, C_i\}$ for every $i$,
it can be represented by a string over $\{A,B,C\}$ of length $k$. The idea of our
construction is that the PFA will mimic rewriting the string $C^2 A^{k-2}$ to the string
$C^2 A^{k-3} B$ with respect to the rewrite system $R$, which consists of the following three rules
\[ BBA \to AAB, \; CBA \to CAB, \; CCA \to CCB.\]
The key argument is that this rewriting is possible, but requires an exponential number of
steps.
This is elaborated in the following lemma, in which we use $\to_R$ for rewriting with respect to
$R$, that is, $u \to_R v$, if and only if $u = u_1 \ell u_2$ and $v = u_1 r u_2$, for strings
$u_1,u_2$ and a rule $\ell \to r$ in $R$. Its transitive closure is denoted by $\to_R^+$.
Just as in the previous section, we write $\fib$ for the standard Fibonacci function. 
It is well-known that $\fib(n) = \Theta(\phi^n)$.

\begin{lemma}
	\label{lemrewrlen}
	For $k \geq 3$, we have $CCA^{k-2} \to_R^+ CCA^{k-3}B$. Furthermore,
	the smallest possible number of steps for rewriting $CCA^{k-2}$ to a string ending in $B$, is
	exactly $ \fib(k)-1$.
\end{lemma}

\begin{proof}
	For the first claim we do induction on $k$. For $k=3$, we have $CCA \to_R CCB$.
	For $k=4$, we have $CCAA \to_R CCBA \to_R CCAB$. For $k>4$, applying the induction
	hypothesis twice, we obtain
	\[CCA^{k-2} \to_R^+ CCA^{k-4}BA \to_R^+ CCA^{k-5}BBA \to_R CCA^{k-3}B.\]
	For the second claim, we define the \emph{weight} $W(u)$ of a string $u = u_1 u_2 \cdots u_k$ over
	$\{A,B,C\}$ of length $k$ by
	$$
	W(u) = \sum_{i : u_i = B} (\fib(i)-1).
	$$
	So every $B$ on position $i$ in
	$u$ contributes $\fib(i)-1$ to the weight, and the other symbols have no weight.
	
	Now we claim that
	$W(v) = W(u)+1$
	for all strings $u,v$ with $u \to_R v$ and $u,v$ only having $C$'s in the first two positions.
	Since the $C$s only occur at positions 1 and 2, by applying $CCA \to
	CCB$, the weight increases by $\fib(3)-1 = 1$ by the creation of $B$ on position 3,
	and by applying $CBA \to CAB$, it increases by $\fib(4)-1 -(\fib(3)-1) = 1$  since $B$ on position 3
	is replaced by $B$ on position 4.
	By applying $BBA \to AAB$, the contributions to the weight $\fib(i)-1$ and $\fib(i+1)-1$ of the two
	$B$s are replaced by $\fib(i+2)-1$ of the new $B$, which is an increase by 1 according to
	the definition of $\fib$.
	
	So this weight increases by exactly 1 at every rewrite step, hence it
	requires exactly $\fib(k)-1$ steps, to go from the initial string  $CCA^{k-2}$ of weight 0 to the
	weight $\fib(k)-1$ of a $B$ symbol on the last position $k$, if that is the only $B$, and more
	steps if there are more $B$s.
\end{proof}

Now we are ready to define the PFA on $Q = \{A_i, B_i, C_i \mid i = 1,\ldots,k\}$ and three
symbols. The three symbols are a start symbol $s$, a rewrite symbol $r$ and a cyclic
shift symbol $c$. The transitions are defined as follows (writing $\bot$ for undefined):
\[
\renewcommand{\arraystretch}{1.4}
\begin{array}{|@{\quad}r@{\,=\,}l@{\qquad}r@{\,=\,}l@{\qquad}r@{\,=\,}l@{\qquad}l@{\quad}|}
\hline
\multicolumn{3}{|@{\quad}r@{\,=\,}}{A_is = B_is} &
\multicolumn{3}{l}{C_i s = C_i,} & \mbox{for $i = 1,2$}, \\[-5pt]
\multicolumn{3}{|@{\quad}r@{\,=\,}}{A_is = B_is} &
\multicolumn{3}{l}{C_i s = A_i,} & \mbox{for $i = 3,\ldots,k$}, \\
\hline
A_1r & \bot, & B_1r & A_1, & C_1r & C_1, & \\[-5pt]
A_2r & \bot, & B_2r & A_2, & C_2r & C_2, & \\[-5pt]
A_3r & B_3, & B_3r & \bot, & C_3r & B_2, & \\[-5pt]
A_ir & A_i, & B_ir & B_i, & C_ir & C_i, & \mbox{for $i = 4,\ldots,k$}, \\
\hline
A_ic & A_{i+1}, & B_ic & B_{i+1}, & C_ic & C_{i+1}, & \mbox{for $i = 1,\ldots,k-1$}, \\[-5pt]
A_kc & A_1, & B_kc & B_1, & C_kc & C_1. & \\
\hline
\end{array}
\vspace{3pt}
\]
A shortest carefully synchronizing word starts by $s$, since $r$ is not defined on all states and
$c$ permutes all states. After $s$, the set of reached states is
$S(CCA^{k-2}) = \{C_1,C_2,A_3,\ldots,A_k\}$. Here, for a string $u = a_1 a_2 \cdots a_k$
of length $k$ over $\{A,B,C\}$, we write $S(u)$ for the set of $k$ states, containing $A_i$
if and only if $a_i = A$, containing $B_i$ if and only if $a_i = B$, and containing $C_i$
if and only if $a_i = C$, for $i = 1,2,\ldots,k$. Note that for $x \in \{A,B,C\}$ and
$v \in \{A,B,C\}^{k-1}$, we have $S(vx)c = S(xv)$, so $c$ performs a cyclic shift on strings
of length $k$.

The next lemma states that the symbol $r$ indeed mimicks rewriting: applied on sets of the shape
$S(u)$, up to cyclic shift it acts as rewriting on $u$ with respect to $R$ defined above.

\begin{lemma}
	\label{lemrewr}
	Let $u$ be a string of the shape $CCw$, where $w \in \{A,B\}^{k-2}$.
	If $u \to_R v$ for a string $v$, then $S(u)c^irc^{k-i} = S(v)$ for some $i<k$.
	
	Conversely, if $u$ does not end in $B$ and there exists an $i$ such that
	$r$ is defined on $S(u)c^i$, then $u \to_R v$ for a string $v$ of the shape
	$CCw$, where $w \in \{A,B\}^{k-2}$.
\end{lemma}

\begin{proof}
	First assume that $u \to_R v$.
	If $u = u_1 BBA u_2$ and $v = u_1 AAB u_2$, then let $i = |u_2| + 3$, so
	\begin{align*}
		S(u) c^i r c^{k-i}
		&= S(u_1 BBA u_2) c^i r c^{k-i} = S(BBA u_2 u_1) r c^{k-i} \\
		&= S(AAB u_2 u_1) c^{k-i} = S(u_1 AAB u_2) = S(v).
	\end{align*}
	If $u = u_1 CBA u_2$ and $v = u_1 CAB u_2$, then again let $i = |u_2| + 3$, so
	\begin{align*}
		S(u) c^i r c^{k-i}
		&= S(u_1 CBA u_2) c^i r c^{k-i} = S(CBA u_2 u_1) r c^{k-i} \\
		&= S(CAB u_2 u_1) c^{k-i} = S(u_1 CAB u_2) = S(v).
	\end{align*}
	Finally, if $u = u_1 CCA u_2$ and $v = u_1 CCB u_2$, then $u_1 = \lambda$ and the result follows
	for $i=0$.
	
	Conversely, suppose that $S(u)c^ir$ is defined. Since $S(u)c^k = S(u)$, we may assume
	that $i < k$ and can write $u = u_1 u_2$, such that $|u_2| = i$.
	Then $S(u)c^i = S(w)$, where $w = u_2 u_1$.
	Write $w = a_1 a_2 \cdots a_k$. Since $S(u_2 u_1)r$ is defined, we
	get $a_1 \neq A$,  $a_2 \neq A$ and $a_3 \neq B$.
	Moreover, $a_1 = a_2 = a_3 = C$ does not occur since $u$ only contains 2 $C$s,
	and $a_1 a_2 = BC$ or $a_2 a_3 = BC$ does not occur since $u$ does not end in $B$.
	The remaining 3 cases are
	$$
	a_1a_2a_3 = BBA, \qquad a_1a_2a_3 = CBA, \qquad \mbox{and} \qquad a_1a_2a_3 = CCA,
	$$
	where $a_1a_2a_3$ is replaced by the corresponding right hand side of the rule by the
	action of $r$. Then in $S(u)c^irc^{k-i}$, the two $C$s are on positions 1 and 2 again,
	and we obtain $S(u)c^ir c^{k-i} = S(v)$ for a string $v$ of the given shape,
	satisfying $u \to_R v$.
\end{proof}

Combining Lemmas \ref{lemrewrlen} and \ref{lemrewr} and the fact that $\fib(n) = \Omega(\phi^n)$,
we obtain the following.

\begin{corollary}
	\label{correwr}
	There is a word $w$ such that $S(CCA^{k-2})w = S(CCA^{k-3}B)$; the shortest word $w$ for which
	$S(CCA^{k-2})w$ is of the shape $S(u) c^i$ for $u$ ending in $B$ has length $\Omega(\phi^k)$.
\end{corollary}

Now we are ready to prove the lower bound:

\begin{lemma}
	\label{lemlb}
	If $w$ is carefully synchronizing, then $|w| = \Omega(\phi^k)$.
\end{lemma}

\begin{proof}
	Assume that $w$ is a shortest carefully synchronizing word. Then we already observed that the first
	symbol of $w$ is $s$, and $w$ yields $S(CCA^{k-2})$ after the first step in the power automaton.
	By applying only $c$-steps and $r$-steps, according to Lemma \ref{lemrewr},
	only sets of the shape $S(u)c^i$ for which $CCA^{k-2} \to_R^+ u$  can be reached, until u ends in $B$.
	In this process, each $r$-step corresponds to a rewrite step. Applying the
	third symbol $s$ does not make sense, since then we go back to $S(CCA^{k-2})$. According
	to Corollary \ref{correwr},
	in the power automaton at least $\Omega(\phi^k)$ steps are required to reach a set
	which is not of the shape $S(u) c^i$. So for reaching a singleton, the total number of steps is
	at least $\Omega(\phi^k)$.
\end{proof}

Note that for the reasoning until now, the definition of $C_3 r = B_2$ did not play a role, and by
$s,r$ all states were replaced by states having the same index. But after the last symbol of $u$ has
become $B$, this $C_3 r = B_2$ will be applied, leading to a subset in which no state of the group
$A_3, B_3, C_3$ occurs any more.
We could have chosen $C_3 r = A_2$ or $C_3 r = C_2$ as well: it is just that $C_3 r = B_2$ makes $r$ injective, like $c$.
Now we arrive at the main result of this section. Optimizations leading to sharper bounds will be presented in Section \ref{sec:asymptotics}.

\begin{proposition}
	There exists a sequence of transitive carefully synchronizing PFAs with three symbols, $n$ states and shortest synchronizing word length $\Omega(\phi^{n/3})$.
\end{proposition}

\begin{proof}
	Let $n=3k+m$ with $m\in\left\{0,1,2\right\}$. Take our PFA on $3k$ states and select $m$ states with more than one ingoing arrow. Split each of them into two states, each inheriting some of the ingoing arrows. This affects the injectivity of $r$ and $c$, but the PFA remains transitive, and the bound for $3k$ states is maintained. The bound was proved in Lemma \ref{lemlb}; it remains to
	prove that the PFA with $3k$ states is synchronizing, that is, it is possible to end up in a
	singleton in the power automaton.
	
	Let $w$ be the word from Corollary \ref{correwr}.
	Since $S(CCA^{k-2})w = S(CCA^{k-3}B)$
	and the number of $c$'s in $w$ is divisible by $k$, we have $C_1 w = C_1$, $C_2 w = C_2$, $A_3 w
	= A_3, \ldots, A_{k-1} w = A_{k-1}$, $A_k w = B_k$. Hence
	\begin{alignat*}{4}
		\{A_1,B_1,C_1\}swcr &=\,& \{C_1\}cr &=\,& \{C_2\} &\subseteq\,& \{A_1,B_1,C_1\}&c,  \\
		\{A_2,B_2,C_2\}swcr &=& \{C_2\}cr &=& \{B_2\} &\subseteq &\{A_2,B_2,C_2\}&, \\
		\{A_i,B_i,C_i\}swcr &=& \{A_i\}cr &=& \{A_{i+1}\} &\subseteq &\{A_i,B_i,C_i\}&c,
		\quad \mbox{for $i=3,4,\ldots,k-1$,} \\
		\{A_k,B_k,C_k\}swcr &=& \{B_k\}cr &=& \{A_1\} &\subseteq &\{A_k,B_k,C_k\}&c.
	\end{alignat*}
	So for all $i \ne 2$, $\{A_i,B_i,C_i\}swcr$ is contained in the cyclic successor $\{A_i,B_i,C_i\}c$
	of $\{A_i,B_i,C_i\}$. $\{A_2,B_2,C_2\}swcr$ is just contained in $\{A_2,B_2,C_2\}$ itself.
	Since for any $i$, one can take the cyclic successor of $\{A_i,B_i,C_i\}$ at most $k-1$ times before
	ending up in $\{A_2,B_2,C_2\}$, we deduce that
	$$
	\{A_i,B_i,C_i\}(swcr)^{k-1} \subseteq \{A_2,B_2,C_2\}
	\quad \mbox{for $i=1,2,\ldots,k$}.
	$$
	As $\{A_2,B_2,C_2\}s = \{C_2\}$, we obtain the carefully synchronizing word $(s w c r)^{k-1}s$
	of the PFA.
\end{proof}

The word $(s w c r)^{k-1}s$ is a lot longer than necessary. In fact, one can
prove that only $\Omicron(k^2)$ $c$-steps and $\Omicron(k)$ $r$-steps and $s$-steps suffice	after $swcr$.

\section{Reduction to Two Symbols}
\label{sectwos}
In this section we construct PFAs with two symbols and exponential shortest carefully synchronizing word length. We do this
by a general transformation to two-symbol PFAs, as was done before, e.g. in \cite{V16}. There a PFA on
$n$ states and $m$ symbols was transformed to a PFA on $mn$ states and two symbols, preserving
synchronization length. In the next theorem,
we improve this resulting number of states to $(m-1)n$ or even less, only needing a mild extra condition.
Using this result, we reduce our 3-symbol PFA with synchronizing length $\Omega(\phi^{n/3})$ to a
2-symbol PFA with synchronizing length $\Omega(\phi^{n/5})$.

\begin{theorem}
	\label{lem2sym}
	Let $P = (Q, \Sigma)$ be a carefully synchronizing PFA with $|Q| = n$, $|\Sigma| = m$, and
	shortest carefully synchronizing word length $f(n)$. Assume $s \in \Sigma$ and $Q' \subseteq Q$
	satisfy the following properties.
	\begin{enumerate}
		\item there is some number $p$ such that all symbols are defined on $Q s^p$ for a complete symbol $s$,
		\item $qs = q$ for all $q \in Q'$, and
		\item $qa = qb$ for all $q \in Q'$ and all $a,b \in \Sigma \setminus \{s\}$.
	\end{enumerate}
	Let $n' = n - |Q'|$. Then there exists a carefully
	synchronizing PFA on $n + n'\bcdot(m-2)$ states and 2 symbols, with shortest carefully synchronizing word
	length at least $f(n)$. The new PFA is deterministic and/or transitive if $P$ is.
\end{theorem}

Note that if $Q' = \emptyset$ then only requirement 1 remains, and the resulting number of states is
$n + n'\bcdot(m-2) = (m-1)n$.

\begin{proof}
	Write $Q = \{1,2,\ldots,n\}$, $Q' = \{n'+1,\ldots,n\}$,
	and $\Sigma = \{s,a_1,\ldots,a_{m-1}\}$.
	Let the states of the new PFA be
	$P_{1,j}$ for $j = 1,\ldots,n$ and
	$P_{i,j}$ for $i = 2,\ldots,m-1$, $j = 1,\ldots,n'$. Define the following two symbols $a,b$ on
	these states:
	\[P_{i,j} a = \begin{cases}
	P_{i+1,j}, & \text{if $i < m-1, j \leq n'$}, \\
	P_{1,j s}, & \text{if $i = m-1, j \leq n'$}, \\
	P_{1,j}, & \text{if $i = 1, j > n'$}.
	\end{cases}
	\qquad
	\begin{array}{cccccc}
	P_{1,1} & \cdots & P_{1,n'} & P_{1,n'+1} & \cdots & P_{1,n} \\
	P_{2,1} & \cdots & P_{2,n'} && \\
	\vdots & & \vdots &&\\
	P_{m-1,1} & \cdots & P_{m-1,n'} &&
	\end{array}
	\]
	and $P_{i,j} b = P_{1,j a_i}$, for all $i = 1,\ldots,m-1$ and $j = 1,\ldots,n$
	for which $P_{i,j}$ exists and $j a_i$ is defined.
	
	If we arrange the states as indicated above, then on the leftmost $n'$ columns,
	$a$ moves the states one step downward if possible, and
	for the bottom row jumps to the top row and
	acts there as $s$. For the remainder of the top row $a$ also acts as $s$
	(which is the identity). On the leftmost $n'$ columns, the symbol $b$ acts
	as $a_i$ on row $i$ and then jumps to the top line. For the remainder of the
	top row, all $a_i$ act in the same way and $b$ acts likewise.
	
	Define
	$\psi(a_i) = a^{i-1} b$ for $i = 1,\ldots,m-1$, and $\psi(s) = a^{m-1}$. Then on the top line
	$\psi(a_i)$ acts in the same way as $a_i$ in the original PFA. Similarly, $\psi(s)$ acts as $s$.
	On any other row, $\psi(s)$ acts as $s$, too. Since every symbol $a_i$ is
	defined on $qs^p$ for every $q\in Q$, we obtain that $\psi(s)^pb=a^{(m-1)p}b$
	is defined on every state and ends up in the top row.
	
	Assume that $w$ is carefully synchronizing in the original PFA.
	Then by the above observations, $a^{(m-1)p} b \psi(w)$ is carefully synchronizing in the new PFA.
	Conversely, any carefully synchronizing word of the new PFA can be written as $\psi(w)a^j$,
	where $0\leq j\leq m-2$ and $\psi(w)$ is a
	concatenation of blocks of the form $\psi(l),l\in\Sigma$. Now note that
	$a^j$ can never synchronize two distinct states in the top row. Therefore,
	$\psi(w)$ synchronizes the top row and consequently $w$ is synchronizing in
	the original PFA. Clearly $|\psi(w)a^j|\geq|w|\geq f(n)$.
\end{proof}

We apply Theorem \ref{lem2sym} to our basic construction with $3k$ states and $m=3$ symbols;
note that $s,c$ are defined on all states and $r$ is defined on $Qs$, so the requirements of
Theorem \ref{lem2sym} hold for $p=1$. As $r$ and $c$ act differently on all states, the only option for
$Q'$ is $Q' = \emptyset$.
Hence we obtain a carefully synchronizing PFA on $(m-1)3k = 6k$ states and two symbols, with
shortest carefully synchronizing word length $\Omega(\phi^k)$.
For $n$ being the number of states of the new PFA, this is $\Omega(\phi^{n/6})$.

However, instead of our three symbols $s,c,r$ we also get careful synchronization on the three
symbols $s,c,rc$ with careful synchronization length of the same order. But then for $i = 4,\ldots,k$ we
have $A_i s = A_i$ and $A_i c = A_i rc$, so we may choose $Q' = \{A_4,\ldots,A_k\}$ in Theorem
\ref{lem2sym}, by which $n' = 3k - (k-3) = 2k+3$, yielding a PFA on two symbols and $5k+3$ states.
This results in the following proposition, where for $n$ not of the shape $5k+3$ we remove up to four states from $Q'$.

\begin{proposition}
	There exists a sequence of transitive carefully synchronizing PFAs with two symbols, $n$ states and shortest synchronizing word length $\Omega(\phi^{n/5})$.
\end{proposition}
This result will be sharpened in the next section as well. 

\section{Main Asymptotic Results}\label{sec:asymptotics}

In this section we discuss some further optimizations. First we extend the number of rewrite rules and then the number of letters in the system. These rewrite systems will be used to construct PFAs on two and three symbols for which we will derive asymptotic lower bounds for the synchronization length.

\subsection{More Rewrite Rules}

For any $h \geq 2$ we define a rewrite system $R_h$ by taking $h+1$ rewrite rules
\begin{equation}
	C^i B^{h-i} A \to C^i A^{h-i} B
\end{equation}
for $i = 0,\ldots,h$. Then it is possible to construct a PFA that mimicks rewriting of the string $C^hA^{k-h}$ to $C^hA^{k-h-1}B$ in the system $R_h$. For $h=2$ this coincides with our construction in Section \ref{secexp}, but for $h > 2$, this gives a better bound. The following lemma gives the number of steps needed. Note that $f_2(i)$ is equal to $\fib(i)-1$.
\begin{lemma}
	For $k\geq h+1$, we have $C^hA^{k-h} \to_{R_h}^+ C^hA^{k-h-1}B$. Furthermore, the smallest possible number of steps for rewriting $C^hA^{k-h}$ in the system $R_h$ to a string ending in $B$ is exactly $f_h(k)$, where $f_h(k)$ satisfies the recursion
	\[
	f_h(k) = \left\{\begin{array}{ll}
	0 & 1\leq k \leq h\\
	1+\sum_{j=1}^h f_h(k-j) \qquad& k\geq h+1 
	\end{array}
	\right.
	\]  
\end{lemma}

\begin{proof} The proof is essentially analogous to the proof of Lemma \ref{lemrewr}. We define the weight $W(u)$ of a string $u = u_1u_2\ldots u_k$ over $\left\{A,B,C\right\}$ by assigning weight $w_i$ to a $B$ on position $i$:
	\[
	W(u)=\sum_{i:u_i=B} w_i.
	\] 
	Other symbols have zero weight. Now we want to choose $w_i$ in such a way that every rewrite step increases the weight of a string by 1. This gives a recursion for $w_i$: to create a $B$ in position $i$, we need $u_j$ to be equal to $B$ or $C$ for all $i-h\leq j \leq i-1$. After that, one extra rewrite step is needed. We start having already $C$'s in positions $1,\ldots,h$. Therefore $w_i$ satisfies
	\[
	w_i = \sum_{j = \max\left\{h+1,i-h\right\}}^{i-1} w_j, 
	\]
	which means that $w_k=f_h(k)$ as defined in the lemma. By construction, to reach a string ending in $B$, exactly $f_h(k)$ rewrite steps are needed. 
	\end{proof}

\subsection{More Rewrite Symbols}

Instead of just having $A$ and $B$ and rewriting the final $A$ in a string into a $B$, we could take $m$ symbols $A^{(1)},\ldots, A^{(m)}$. For convenience we will sometimes denote $A^{(1)}$ by $A$ and $A^{(m)}$ by $B$. We take $(h+1)(m-1)$ rewrite rules
\begin{equation}\label{eq:rewrite}
	C^iB^{h-i}A^{(t)} \to C^iA^{h-i}A^{(t+1)},
\end{equation}
for $i = 0,\ldots, h$ and $t = 1,\ldots, m-1$. In this rewrite system $R_{h,m}$ the goal is to rewrite the string $C^hA^{k-h}$ into a string ending in $B$. 
\begin{lemma}\label{lem:moresym}
	For $k\geq h+1$, we have $C^hA^{k-h} \to_{R_{h,m}}^+ C^hA^{k-h-1}B$. Furthermore, the smallest possible number of steps for rewriting $C^hA^{k-h}$ in the system $R_{h,m}$ to a string ending in $B$ is exactly $f_{h,m}(k)$, where $f_{h,m}(k)$ satisfies the recursion
	\[
	f_{h,m}(k) = \left\{\begin{array}{ll}
	0 & 1\leq k \leq h\\
	(m-1)\cdot\left(1+\sum_{j=1}^h f_{h,m}(k-j)\right) \qquad& k\geq h+1 
	\end{array}
	\right.
	\]  
\end{lemma}
\begin{proof}
	We define the weight $W(u)$ of a string $u_1u_2\ldots u_k$ by assigning weights to the symbols $A^{(t)}$ for $t\geq 2$:
	\[
	W(u) = \sum_{t=2}^{m}\sum_{i:u_i = A^{(t)}} w_{i,t},
	\]
	where $w_{i,t}$ is the weight of $A^{(t)}$ on position $i$. The symbols $C$ and $A^{(1)}=A$ have zero weight. Again weights will be chosen such that every rewrite step increases the weight of a string by 1. Before we can replace a symbol $A^{(t)}$ in position $i$ by $A^{(t+1)}$, we need $u_j$ to be equal to $A^{(m)}=B$ or $C$ for all $i-h\leq j\leq i-1$. After that, one extra rewrite step is needed. To replace a symbol $A^{(1)}$ in position $i$ by $A^{(t)}$, this has to be repeated $t-1$ times. Therefore, we find the following recursion
	\[
	w_{i,t} = (t-1)\left(1+\sum_{j = \max\left\{h+1,i-h\right\}}^{i-1} w_{j,m}\right).
	\]
	Then for all $k\geq h+1$, $W\left(C^hA^{k-h-1}B\right) = w_{k,m} = f_{h,m}(k)$.
	\end{proof}

\subsection{Construction of the PFA on Three Symbols}\label{sec:constr}

Using the rewrite rules (\ref{eq:rewrite}), we can construct a PFA $P_{h,m}^n$ on $n=(m+1)k$ states $A_{i}^{(1)},\ldots, A_{i}^{(m)}, C_i$ for $i = 1,\ldots, k$. As before, we have a start symbol $s$, a rewrite symbol $r$ and a cyclic shift symbol $c$. Let $X$ denote any of the letters $A^{(1)},\ldots, A^{(m)}, C$. Then
\begin{eqnarray}
	X_is  &= \left\{\begin{array}{ll}
		C_i & i = 1,\ldots,h\\
		A_i^{(1)}\qquad & \textrm{otherwise}
	\end{array}\right.\qquad
	X_ic &= \left\{\begin{array}{ll}
		X_1 & i = k\\
		X_{i+1}\qquad & \textrm{otherwise}
	\end{array}\right.
\end{eqnarray}
Rewriting takes place in the states with indices $1\leq i\leq h+1$. We define it on $(m+1)$-tuples with index $i$ by 
\[
\left(A_i^{(1)}r,\ldots,A_i^{(m)}r, C_ir\right) = \left\{\begin{array}{lll}
\left(\bot,\ldots,\bot,A_i^{(1)},C_i\right)&i = 1,\ldots,h\\
\left(A_i^{(2)},\ldots,A_i^{(m)},\bot,A_{i-1}^{(m)}\right)\qquad&i = h+1\\
\left(A_i^{(1)},\ldots,A_i^{(m)}, C_i\right) & \textrm{otherwise.}
\end{array}\right.
\]
\begin{lemma}\label{lem:constr}
	The PFA $P_{h,m}^n$ is carefully synchronizing and the shortest synchronizing word has length at least $f_{h,m}(n/(m+1))$. 
\end{lemma}
\begin{proof}
	Let $Q$ be the state set of the PFA $P_{h,m}^n$. For a string $u = u_1\ldots u_k$ over $\left\{A^{(1)},\ldots,A^{(m)},C\right\}$, we define $S(u)\subseteq Q$ in such a way that $X_i\in S(u)$ if and only if $u_i=X$. Then $Qs = S(C^hA^{k-h})$. Every application of the symbol $r$ to a set $S(u)$ corresponds to application of a rewrite rule to $u$. As long as the string does not end in $B$, no other changes are possible, except for cyclic shifts and resetting to $C^hA^{k-h}$. To reach the set $S(C^hA^{k-h-1}B)$, we need at least a word of length $f_{h,m}(k) = f_{h,m}(n/(m+1))$. 
	
	To see that $P_{h,m}^n$ is synchronizing, let $w$ be such that $Qsw = S(C^hA^{k-h-1}B)$,
	and let \[Q_i = \left\{A_i^{(1)},\ldots,A_i^{(m)},C_i\right\}.\] 
	Then $Q_i swcr\subseteq Q_{(i+1) \bmod k}$ for all $i\neq h$,
	and $Q_h swcr \subseteq Q_h$. Consequently, $Q(swcr)^{k-1}s = \{C_h\}$ so that 
	$(swcr)^{k-1}s$ is synchronizing.
	\end{proof}

\subsection{Asymptotic Lower Bound for PFAs on Three Symbols}

\begin{theorem} There exists a sequence of transitive carefully synchronizing PFAs with three symbols, $n$ states and shortest carefully synchronizing word length \[\Omega\left(2^{\frac{2}{5}n-\log_2(n)}\right) = \Omega\left(\frac{2^{2n/5}}{n}\right).\]
\end{theorem}

\begin{proof} As before, we can reduce to the case where $m+1\mid n$. For this case, we analyze the recursion of Lemma \ref{lem:moresym} and choose $h\geq 2$ and $m\geq 2$ dependent on $n$ in such a way that $f_{h,m}(n/(m+1))$ is maximal. First note that the recursive equations can be rewritten to a homogeneous system, by taking $g_{h,m}(k) = f_{h,m}(k)+\frac{m-1}{(m-1)h-1}$:
	\begin{equation}\label{eq:recursiong}
		g_{h,m}(k) = \left\{\begin{array}{ll}
			\frac{m-1}{(m-1)h-1}& k = 1,\ldots,h\\
			(m-1)\sum_{j=1}^h g_{h,m}(k-j)\qquad & k\geq h+1
		\end{array}\right.
	\end{equation}
	The case $m=h=2$ gives Fibonacci's sequence. The general homogeneous recurrence relation has characteristic equation
	\[
	x^h = (m-1)\left(x^{h-1}+x^{h-2}+\ldots+1\right) = (m-1)\cdot\frac{x^h-1}{x-1},
	\]
	provided $x\neq 1$. It can be rewritten as $m-x = (m-1)x^{-h}$, having a solution close to $m$. Indeed, for given $\varepsilon>0$, we can choose $h$ large enough so that there is a solution $\phi_{h,m}$ satisfying $m-\varepsilon\leq\phi_{h,m}<m$. This gives exponential growth with rate at least $m-\varepsilon$, which we will prove by induction. For $k=h+1,\ldots,2h$, we have
	\begin{equation}\label{eq:hyp}
		g_{h,m}(k)\geq f_{h,m}(k) = (m-1)m^{k-h-1}\geq \frac{m-1}{(m-\varepsilon)^{h+1}}\cdot(m-\varepsilon)^k.  
	\end{equation}
	For $k>2h$, assuming the above inequality for $k'<k$, we obtain
		\begin{align*}
		g_{h,m}(k) &= (m-1) \sum_{j=1}^h g_{h,m}(k-j) \\
		&\ge \frac{(m-1)^2}{(m-\varepsilon)^{h+1}}\cdot
		      (m-\varepsilon)^{k-1}\cdot
		      \sum_{j=1}^h(m-\varepsilon)^{1-j} \\
		&\ge \frac{(m-1)^2}{(m-\varepsilon)^{h+1}}\cdot
		      \frac{(m-\varepsilon)^{k}}{m-\varepsilon}\cdot
		      \frac{1-(m-\varepsilon)^{-h}}{1-(m-\varepsilon)^{-1}}\\
		& = \frac{m-1}{(m-\varepsilon)^{h+1}}\cdot
		      (m-\varepsilon)^{k}\cdot
		      \frac{(m-1)-(m-1)(m-\varepsilon)^{-h}}{(m-1)-\varepsilon}.
	\end{align*}
	In order to prove (\ref{eq:hyp}) for all $k > h$, the second fraction on the right hand side must be at least $1$, which is equivalent to
	$$
	(m-1)(m-\varepsilon)^{-h} \le \varepsilon.
	$$
	So we take
	$$
	h = \Big\lceil \log_{m-\varepsilon} \Big( \frac{m-1}{\varepsilon} \Big) \Big\rceil.
	$$
	This implies $(m-\varepsilon)^{h-1}< (m-1)/\varepsilon$, which we substitute in (\ref{eq:hyp}) to get for $k>h$
	\begin{equation}\label{eq:growthrate}
	g_{h,m}(k) \geq \frac{m-1}{(m-\varepsilon)^2}\cdot\frac{\varepsilon}{m-1}\cdot(m-\varepsilon)^k \geq \frac{1}{m^2}\cdot\varepsilon\cdot(m-\varepsilon)^k.
	\end{equation}
	The PFA that has to be constructed has $n=(m+1)k$ states, so to find the growth in $n$, 
	we substitute $k=\frac{n}{m+1}$ in \eqref{eq:growthrate}. The following indirect argument shows that this is a valid choice, i.e. that $k>h$. If $k \le h$, then the right hand side of \eqref{eq:growthrate} would be less than $1$. Since our choice of $k$ will lead to a lower bound greater than $1$, we deduce that $k>h$. The best choice for $m$ is $m=4$, since this maximizes the growth rate $m^{1/(m+1)}$. This choice gives
	\[
	g_{h,m}\left(\frac{n}{m+1}\right) \geq \frac{1}{16}\cdot\varepsilon\cdot(4-\varepsilon)^{n/5}.
	\]
	Finally, we want to choose $\varepsilon$. The right hand side has a maximum at $\varepsilon = 20/(n+5)$. Note that this means that we rewrite a string of length $k$ by substituting blocks of length $h+1$ proportional to $\log(k)$. This choice leads to
	\[
	g_{h,m}\left(\frac{n}{m+1}\right) \geq \frac{20}{16(n+5)}\left(4-\frac{20}{n+5}\right)^{n/5} \geq \frac{5}{4(n+5)}\cdot 4^{n/5} \cdot\left(1-\frac5{n}\right)^{n/5},
	\]
	in which the last factor is bounded from below by a positive constant. Therefore 
	\[
	g_{h,m}\left(\frac{n}{m+1}\right) = \Omega\left(\frac{4^{n/5}}{n}\right) = \Omega\left(2^{\frac{2}{5}n-\log_2(n)}\right).
	\]
	Since $f_{h,m}\big(n/(m+1)\big)$ has the same growth rate, Lemma \ref{lem:constr} gives the result.\end{proof}

\subsection{Construction of the Binary PFA}\label{sec:constrbin}

To obtain an asymptotic lower bound for binary PFAs, we will use the reduction technique from Section \ref{sectwos}. Before doing so, we will slightly tune the construction of $P_{h,m}^n$ so that we obtain a bigger set $Q'$ in Theorem \ref{lem2sym}. Let $\tilde P_{h,m}^n$ be the PFA with symbols $c$, $rc$ and a modified start symbol $s'$, defined by
\[
\begin{array}{ll}
X_is' = C_i&i=1,\ldots,h\\
A_i^{(t)}s' = A_i^{(t)}\textrm{\ and\ }C_is' = A_i^{(1)}\qquad &i = h+2,\ldots,k-4;\ t=1,\ldots,m\\
X_is' = A_i^{(1)}&i = h+1 \textrm{\ and\ } i = k-3,\ldots,k,
\end{array}
\]
where $X$ stands for any of the symbols $C,A^{(1)},\ldots,A^{(m)}$.
\begin{lemma}
	The automaton $\tilde P_{h,m}^n$ is carefully synchronizing and its shortest synchronizing word has length at least $\Omega\left(2^{\frac25n-\log_2(n)}\right)$.
\end{lemma}
\begin{proof}
	Since $s = (cs')^k$, we deduce that $\tilde P_{h,m}^n$ is synchronizing. Now the set $Qs'$ corresponds to a collection of strings of length $k$. More precisely, all strings of the form $u = C^hAu_{h+2}\ldots u_{k-4}A^4$ with $u_j\in\left\{A^{(1)},\ldots,A^{(m)}\right\}$ for $j=h+2,\ldots,k-4$. In this collection, the string with maximal weight is $u_{\rm{max}} := C^hAB^{k-h-5}A^4$. The number of steps to synchronize $\tilde P_{h,m}^n$ is at least the number of steps needed to rewrite $u_{\rm{max}}$ into a string ending in $B$. To show that this is of the same order as rewriting a string of weight zero, it suffices to show that
	\begin{equation}\label{eq:suffi}
		W(u_{\rm{max}}) \leq \frac12 W\left(C^hA^{k-h-1}B\right).
	\end{equation}
	First we prove that
	\begin{equation}\label{eq:W}
		W\left(C^hB^iA^{k-h-i}\right) \leq W\left(C^hA^{i+1}BA^{k-h-i-2}\right), \qquad 0\leq i\leq k-h-2.
	\end{equation}   
	For $i=0$ and $i=1$ this is clear by construction of the weights. By induction it then follows that 
	\begin{align*}
		W(C^hB^{i+2}A^{k-h-i-2})
		&= W(C^hB^{i}A^{k-h-i}) + W(C^hA^{i}B^2A^{k-h-i-2}) \\
		&\le W(C^hA^{i+1}BA^{k-h-i-2}) + W(C^hA^{i+2}BA^{k-h-i-3}) \\
		&= W(C^hA^{i+1}B^2A^{k-h-i-3}) \\
		&\le W(C^hA^{i+3}BA^{k-h-i-4}),
	\end{align*}
	so (\ref{eq:W}) holds if $0\leq i\leq k-h-2$. Taking $i=k-h-4$, we obtain
	\begin{align*}
		W(u_{\rm{max}}) & \leq W(C^hB^{k-h-4}A^4)\leq W(C^hA^{k-h-3}BA^2)\leq\frac12W(C^hA^{k-h-1}B), 
	\end{align*}
	proving (\ref{eq:suffi}), so that we get a lower bound of the same order as for $P_{h,m}^n$.
	\end{proof}
\subsection{Asymptotic Lower Bound for Binary PFAs}
We will apply Theorem \ref{lem2sym} to the sequence $\tilde P_{h,m}^n$ to derive an asymptotic lower bound for binary PFAs. 
\begin{theorem} There exists a sequence of carefully synchronizing PFAs with two symbols, $n$ states and shortest carefully synchronizing word length \[\Omega\left(2^{\frac{1}{3}n-\frac{3}{2}\log_2(n)}\right) = \Omega\left(\frac{2^{n/3}}{n\sqrt{n}}\right).\]
\end{theorem}
\begin{proof}
	Consider the PFA $\tilde P_{h,m}^n$. We check the conditions for Theorem \ref{lem2sym}. The symbol $s'$ is complete and all symbols are defined on $Qs'$. Define the set of states $Q'$ by
	\[
	Q' = \left\{C_1,\ldots,C_h\right\}\cup\big\{A_i^{(t)}\big|h+2\leq i\leq k-4,\ 1\leq t\leq m\big\},
	\]
	fulfilling all conditions of Theorem \ref{lem2sym}. The reduction in this case gives a binary PFA on $N = 2n-|Q'| = 2(m+1)k-h-m(k-h-5)=(m+2)k+(m-1)h+5m$ states so that
	\[
	k = \frac{N-(m-1)h-5m}{m+2}.
	\]
	For $g_{h,m}(k)$ we still have the lower bound $\varepsilon\cdot(m-\varepsilon)^k/m^2$ as in (\ref{eq:growthrate}). This time the main order term is $m^{N/(m+2)}$, which is again maximized by taking $m=4$. For $\varepsilon$, the best choice is $\varepsilon = 24/(N+6)$, which means $h=\log_{m-\varepsilon}(N)+O(1)$. 	Finally, we conclude that the length of the shortest synchronizing word for $\tilde P_{h,m}^n$ is bounded by
	\begin{align*}
	\Omega\left(\varepsilon\cdot(m-\varepsilon)^k\right) &= \Omega\left(\varepsilon\cdot(4-\varepsilon)^{N/6}\cdot(4-\varepsilon)^{-\log_{4-\varepsilon}(N)/2}\right)\\
	&=\Omega\left(\frac{4^{N/6}}{N\sqrt{N}}\right) =  \Omega\left(2^{\frac{N}{3}-\frac{3}{2}\log_2(N)}\right).
	\end{align*}
	If the number of states is not of the form $(m+2)k+(m-1)h+5m = 6k+3h+20$, then we remove some
	states from $Q'$, just as before.
\end{proof}  

\subsection{PFAs with a Single Undefined Transition}

The PFA $P_{h,m}^n$ as defined in Section \ref{sec:constr} has $h(m-1)+1$ undefined 
transitions. In this section we present a variation on the theme, showing that a 
single undefined transition suffices to get exponential synchronizing word lengths. 
First note that the recursion in Lemma \ref{lem:moresym} gives exponential growth 
for $h=1$, provided $m\geq 3$. In this case, the recursion reduces to 
$f_{1,m}(k) = (m-1)(1+f_{1,m}(k-1))$ with $f_{1,m}(1)=0$. By a straightforward 
inductive argument, it follows that 
\[
f_{1,m}(k) = \frac{m-1}{m-2}\left((m-1)^k-1\right).
\]
We will use the system (\ref{eq:rewrite}) to rewrite 
$C\left(A^{(1)}\right)^{k-1}$ into a string ending in $A^{(m)}$. We extend 
the rewrite system for $h=1$ and $m\geq 3$ with $m-2$ letters 
$A^{(m+1)},\ldots,A^{(2m-2)}$. We also extend the set of rewrite rules to
\[
A^{(s)}A^{(t)}\to A^{(s-(m-1))}A^{(t+1)}\qquad\textrm{and}\qquad C A^{(t)}\to C A^{(t+1)}
\]
for $t = 1,\ldots,2m-3$ and $s \geq m$. Furthermore, we close the system cyclically by
\[
A^{(s)}A^{(2m-2)}\to A^{(s-(m-1))}A^{(1)}\qquad\textrm{and}\qquad C A^{(2m-2)}\to C A^{(1)}.
\]
for $s \ge m$. Just as before, for any $t \ge 1$, the weight of $A^{(t)}$ on some 
position is $t-1$ times the weight of $A^{(2)}$ on the same position, and we
see that the new rewrite rules 
either decrease the weight of a string or increase it by at most 1. 
So these extensions will not reduce the number of rewrite steps needed. 

Now we build a PFA on $n=(2m-1)k$ states $A_i^{(0)},A_i^{(1)},\ldots,A_i^{(2m-2)}$ 
for $i=1,\ldots,k$ mimicking this rewrite system. Just as before, we make the rewrite symbol $r$ injective, but this time it makes the construction slightly more complicated than necessary. The idea will be that a letter 
$A^{(t)}$ on position $i$ in the string corresponds to the set of states
\[
S_i^{(t)} := \left\{\begin{array}{ll}
\left\{A_i^{(t)},\ldots,A_i^{(t+m-2)}\right\}&1\leq t\leq m,\\
\left\{A_i^{(1)},\ldots,A_i^{(t-m)},A_i^{(t)},\ldots,A_i^{(2m-2)}\right\}
\quad&m+1\leq t\leq 2m-2.
\end{array}
\right.
\]
Furthermore, the letter $C$ on position $i$ will correspond to either
$$
S_i^{(0)} := \left\{A_i^{(0)},A_i^{(1)},\ldots,A_i^{(m-2)}\right\} \qquad \mbox{or} \qquad 
\bar{S}_i^{(0)} := \left\{A_i^{(0)},A_i^{(m)},\ldots,A_i^{(2m-3)}\right\}.
$$
We define the start symbol $s$ and the cyclic shift symbol $c$ by
\begin{gather*}
A_i^{(t)}s = \left\{\begin{array}{ll}
A_i^{(t \bmod (m-1))}\quad& (m-1) \nmid t \mbox{ or } i = 1\\
A_i^{(m-1)}& \textrm{otherwise}
\end{array}\right.\qquad
A_i^{(t)}c = \left\{\begin{array}{ll}
A_1^{(t)}&i=k\\
A_{i+1}^{(t)}\quad& \textrm{otherwise}
\end{array}\right.
\end{gather*}
With this definition of the start symbol $s$, we have $Qs = S_1^{(0)}\cup\bigcup_{i=2}^kS_i^{(1)}$, 
representing the string $A^{(0)}\left(A^{(1)}\right)^{k-1}$.
The (injective) rewrite symbol is defined by 
\[
\left(A_i^{(0)}r,\ldots,A_i^{(2m-2)}r\right) = 
\left\{\begin{array}{ll}
\left(A_1^{(0)},A_1^{(m)},\ldots,A_1^{(2m-3)},\bot,A_1^{(1)},\ldots,A_1^{(m-1)}\right)\quad&i = 1\\
\left(A_1^{(2m-2)},A_2^{(2)},\ldots,A_2^{(2m-2)},A_2^{(1)}\right)&i = 2\\
\left(A_i^{(0)},\ldots,A_i^{(2m-2)}\right)&i\geq 3
\end{array}\right.
\]
This implies that $r$ acts on the sets $S_i^{(t)}$ for $i=1,2$ as
\[
S_1^{(t)}r = \left\{\begin{array}{ll}
\bar{S}_1^{(0)} &t=0\\ \bot & 1\leq t\leq m-1\\S_1^{(t-(m-1))}\quad&m\leq t\leq 2m-2,
\end{array}
\right.\qquad
S_2^{(t)}r = \left\{\begin{array}{ll}S_2^{(t+1)}\quad & 1\leq t\leq 2m-3 \\
S_2^{(1)}&t=2m-2,\end{array}\right.
\]
Since $\bar{S}_1^{(0)}r = S_1^{(0)}$ in addition, the states with indices $i=1$ and $i=2$ 
exactly mimic the rewrite rules. 
The action of $r$ onto $S_2^{(0)}$ or $\bar{S}_2^{(0)}$ does not give a set of the form 
$S_i^{(t)}$, but can only be applied after reaching a string ending in $A^{(t)}$ for 
some $t \ge m$. For $i\geq 3$, we have $S_i^{(t)}r = S_i^{(t)}$ for every $t$, and 
$\bar{S}_i^{(0)}r = \bar{S}_i^{(0)}$. 

This construction leads to the following theorem.

\begin{theorem} 
	There exists sequences of carefully synchronizing PFAs with only one undefined 
	transition and shortest carefully synchronizing word length 
	\begin{itemize}
		\item $\Omega(3^{n/7})$ for PFAs on three symbols,
		\item $\Omega(2^{n/5})$ for PFAs on two symbols,
	\end{itemize}
	where $n$ is the number of states. 
\end{theorem}

\begin{proof} First we argue that the PFA constructed above is synchronizing.
	Let $w$ be a word to rewrite $C(A^{(1)})^{k-1}$ into a string ending in $A^{(m)}$.
	Write $Q_i$ for the set of states with subindex $i$ for each $i$. 
	Then $Q_i c^{k-1}swc \subseteq Q_i$ for all $i$. Furthermore, $r$ is defined
	on $Q c^{k-1}swc$, so let $v = c^{k-1}swcr$. Then $Q_i v \subseteq Q_i$ for all
	$i \neq 2$.
	
	To investigate $Q_2 v$, we group states modulo $m-1$, by defining 
	$B_i^{(0)} = \{ A_i^{(0)}, A_i^{(m-1)},  A_i^{(2m-2)} \}$, and 
	$B_i^{(t)} = \{ A_i^{(t)}, A_i^{(t+m-1)}\}$ for all $t \ne 0$.
	Then $B_1^{(t)} sw \subseteq B_1^{(t)}$, so
	$$
	B_2^{(t)} v \subseteq B_2^{(t)} r \subseteq B_2^{((t+1) \bmod (m-1))} 
	$$
	for all $t \ne 0$, and
	$$
	B_2^{(0)} v \subseteq \{A_2^{(0)}\}r \subseteq B_1^{(0)}.
	$$
	Consequently, $Q_2 v^{m-1} \subseteq Q_1$. 
	
	Furthermore, it follows by 
	induction that $Q \,\allowbreak (v^{m-1}c)^{k-1} \subseteq Q_1$. From 
	$B_k^{(0)} w \subseteq B_k^{(0)}$, we infer that $B_1^{(0)} v \subseteq
	B_1^{(0)}$. Hence $Q\, (v^{m-1}c)^k \subseteq B_1^{(0)}$. 
	As $B_1^{(0)} s = \{A_1^{(0)}\}$, we conclude
	that $(v^{m-1}c)^k s$ is a synchronizing word.
	So we have a synchronizing $n$-state PFA with synchronizing word length 
	$$
	\Omega((m-1)^k) = \Omega((m-1)^{\frac{n}{2m-1}})
	$$
	The best choice is $m=4$, leading to the lower bound $\Omega(3^{n/7})$.
	If $n$ is not of the form $7k$, then we can split up states just as before,
	but we must not split up state $A_1^{(m-1)}$.
	
	For the binary construction with a single undefined transition, 
	we proceed in the spirit of Section \ref{sec:constrbin}. 
	We take symbols $c$ and $rc$ and use an adapted start symbol 
	$s'$ such that $A_i^{(t)}s' = A_i^{(t)}$ for $1\leq t\leq 2m-2$ and 
	$3\leq i\leq k-4$. Now we can apply Theorem \ref{lem2sym} with a set 
	$Q'$ of size $(2m-2)(k-6)+O(1)$. This gives a PFA on $n=2mk$ states and a 
	lower bound 
	\[
	\Omega((m-1)^k) = \Omega((m-1)^{\frac{n}{2m}}).
	\] 
	The best choice is $m=5$, giving $\Omega(4^{n/10}) = \Omega(2^{n/5})$.
	\end{proof}

\section{Conclusions}
\label{secconcl}
For every $n$, we constructed PFAs on $n$ states and two or three symbols for which careful 
synchronization is forced to mimic rewriting with respect to a string rewrite system. 
These systems require an exponential number of steps to reach a string of a particular shape. 
The resulting exponential lengths are much larger than the cubic upper bound for synchronization 
of DFAs. We show that for $n=4$ the shortest synchronization length for a PFA already can exceed 
the maximal shortest synchronization length for a DFA. 

For $n\leq 7$ we found greatest possible shortest synchronization lengths, both for DFAs and PFAs,
where for DFAs until now this was only fully investigated for $n \leq 4$, that is, by not assuming 
any bound on the number of symbols. For these $n$, we identify PFAs reaching the maximal length. 
These extreme cases require up to eight symbols, where for DFAs the maximal lengths are generally 
attained by binary examples.

Besides the proof of Theorem \ref{th:Pn}, several results which are related to those in this paper 
were not selected in this paper. One of those results is a generalization of the class $P_n$ in 
Section \ref{secbinpfa}. The other results have been gathered in \cite{debondt18}.

{\bf Acknowledgement:}
We thank Stijn Cambie for his contribution to the proof of Theorem \ref{th:Pn}.

\bibliography{ref}

\begin{thebibliography}{10}

\bibitem{AVG12}
D.~S. Ananichev, M.~V. Volkov and V.~V. Gusev, Primitive digraphs with large
  exponents and slowly synchronizing automata, {\em Zap. Nauchn. Sem.
  S.-Peterburg. Otdel. Mat. Inst. Steklov. (POMI)} {\bf 402}(Kombinatorika i
  Teoriya Grafov. IV)  (2012)  9--39, 218.

\bibitem{C64}
J.~{\v{C}}ern\'y, Pozn\'amka k homog\'ennym experimentom s kone\v{c}n\'ymi
  automatmi, {\em Matematicko-fyzik\'alny \v{c}asopis, Slovensk. Akad. Vied}
  {\bf 14}(3)  (1964)  208--216.

\bibitem{B17}
M.~de~Bondt, Fast algorithms for anti-distance matrices as a generalization of
  {B}oolean matrices, available at {\tt https://arxiv.org/abs/1705.08743}
  (2017).

\bibitem{debondt18}
M.~de~Bondt, Subset synchronization of {DFA}s and {PFA}s, and some other
  results, Available at {\tt http://arxiv.org/abs/1807.04661}  (2018).

\bibitem{BDZ17}
M.~de~Bondt, H.~Don and H.~Zantema, {DFAs} and {PFAs} with long shortest
  synchronizing word length, {\em Developments in Language Theory\/},  eds.
  {\relax Charlier}.~\'E., {\relax Leroy}.~J. and {\relax Rigo}.~M. {\em
  Lecture Notes in Computer Science} {\bf 10396}, (Springer, Cham, 2017).

\bibitem{BDZ16}
M.~de~Bondt, H.~Don and H.~Zantema, Slowly synchronizing automata with fixed
  alphabet size, \emph{Information and Computation}, to appear. Available at
  {\tt https://arxiv.org/abs/1609.06853}  (2017).

\bibitem{DZ17}
H.~Don and H.~Zantema, Finding {DFA}s with maximal shortest synchronizing word
  length, {\em Language and Automata Theory and Applications\/},  eds. {\relax
  Drewes}.~F., {\relax Mart\'in-Vide}.~C. and {\relax Truthe}.~B. {\em Springer
  Lecture Notes in Computer Science} {\bf 10168}, (Springer, Cham, 2017).

\bibitem{dzyga_et_al:LIPIcs:2017:8122}
M.~Dzyga, R.~Ferens, V.~V. Gusev and M.~Szyku{\l}a, {Attainable Values of Reset
  Thresholds}, {\em 42nd International Symposium on Mathematical Foundations of
  Computer Science (MFCS 2017)\/},  eds. K.~G. Larsen, H.~L. Bodlaender and
  J.-F. Raskin {\em Leibniz International Proceedings in Informatics (LIPIcs)}
  {\bf 83}, (Schloss Dagstuhl--Leibniz-Zentrum fuer Informatik, Dagstuhl,
  Germany, 2017), pp. 40:1--40:14.

\bibitem{frankl}
P.~Frankl, An extremal problem for two families of sets, {\em European Journal
  of Combinatorics} {\bf 3}  (1982)  125--127.

\bibitem{MR3746507}
B.~Gerencs\'er, V.~V. Gusev and R.~M. Jungers, Primitive sets of nonnegative
  matrices and synchronizing automata, {\em SIAM J. Matrix Anal. Appl.} {\bf
  39}(1)  (2018)  83--98.

\bibitem{K01}
J.~Kari, A counterexample to a conjecture concerning synchronizing words in
  finite automata, {\em EATCS Bulletin} {\bf 73}  (2001)  146--147.

\bibitem{10.1007/978-3-319-40946-7_15}
A.~Kisielewicz, J.~Kowalski and M.~Szyku{\l}a, Experiments with synchronizing
  automata, {\em Implementation and Application of Automata\/},  eds. Y.-S. Han
  and K.~Salomaa (Springer International Publishing, Cham, 2016), pp. 176--188.

\bibitem{Mar08}
P.~V. Martyugin, Lower bounds for the length of the shortest carefully
  synchronizing words for two- and three-letter partial automata, {\em
  Diskretn. Anal. Issled. Oper.} {\bf 15}(4)  (2008)  44--56, 99.

\bibitem{M10}
P.~V. Martyugin, A lower bound for the length of the shortest carefully
  synchronizing words, {\em Russian Mathematics (Iz. VUZ)} {\bf 54}(1)  (2010)
  46--54.

\bibitem{10.1007/978-3-642-31606-7_24}
P.~V. Martyugin, Synchronization of automata with one undefined or ambiguous
  transition, {\em Implementation and Application of Automata\/},  eds.
  N.~Moreira and R.~Reis (Springer Berlin Heidelberg, Berlin, Heidelberg,
  2012), pp. 278--288.

\bibitem{10.1007/978-3-642-38536-0_7}
P.~V. Martyugin, Careful synchronization of partial automata with restricted
  alphabets, {\em Computer Science -- Theory and Applications\/},  eds. A.~A.
  Bulatov and A.~M. Shur (Springer Berlin Heidelberg, Berlin, Heidelberg,
  2013), pp. 76--87.

\bibitem{pin}
J.-E. Pin, On two combinatorial problems arising from automata theory, {\em
  Annals of Discrete Mathematics} {\bf 17}  (1983)  535--548.

\bibitem{R08}
A.~Roman, A note on \v{C}ern{\'y} conjecture for automata with 3-letter
  alphabet, {\em Journal of Automata, Languages and Combinatorics} {\bf 13}(2)
  (2008)  141--143.

\bibitem{rystsov}
I.~Rystsov, Asymptotic estimate of the length of a diagnostic word for a finite
  automaton, {\em Cybernetics} {\bf 16}(2)  (1980)  194--198.

\bibitem{szykula18}
M.~Szyku{\l}a, {Improving the Upper Bound on the Length of the Shortest Reset
  Word}, {\em 35th Symposium on Theoretical Aspects of Computer Science (STACS
  2018)\/},  eds. R.~Niedermeier and B.~Vall{\'e}e {\em Leibniz International
  Proceedings in Informatics (LIPIcs)} {\bf 96}, (Schloss
  Dagstuhl--Leibniz-Zentrum fuer Informatik, Dagstuhl, Germany, 2018), pp.
  56:1--56:13.

\bibitem{T06}
A.~N. Trahtman, An efficient algorithm finds noticeable trends and examples
  concerning the \v{C}ern\'y conjecture, {\em Mathematical Foundations of
  Computer Science 2006: 31st International Symposium, MFCS 2006\/},  eds.
  R.~Kr{\'a}lovi{\v{c}} and P.~Urzyczyn (Springer Berlin Heidelberg, 2006), pp.
  789--800.

\bibitem{volkov}
M.~Volkov, Synchronizing automata and the \v{C}ern\'y conjecture, {\em
  Proceedings of LATA\/},  {\em Springer LNCS} {\bf 5196}  (2008), pp. 11--27.

\bibitem{V16}
V.~Vorel, Subset synchronization and careful synchronization of binary finite
  automata, {\em Int. J. Found. Comput. Sci.} {\bf 27}(5)  (2016)  557--578.

\end{thebibliography}

\end{document}